\documentclass{amsart}

\usepackage{amsmath,amsfonts,amssymb}
\usepackage[pdftex]{graphicx}
\usepackage{mathptmx}
\usepackage{natbib}
\usepackage{a4wide}

\newcommand{\bigO}{O}
\newcommand{\e}{\mathrm{e}}
\renewcommand{\d}{\mathrm{d}}

\newcommand{\bydef}{:=}
\newcommand{\eps}{\varepsilon}
\newcommand{\ini}{{\mathrm{I}}}
\newcommand{\R}{\mathbb{R}}
\newcommand{\Ind}{\mathcal{I}}

\newtheorem{thm}{Theorem}
\newtheorem{rmk}{Remark}
\newtheorem{lem}{Lemma}
\newtheorem{cor}{Corollary}
\newtheorem{prop}{Proposition}
\newtheorem{deff}{Definition}
\newcommand{\slf}{S}
\newcommand{\ad}{\mathrm{ad}}

\newcommand{\bx}{\mathbf{x}}
\newcommand{\by}{\mathbf{y}}
\newcommand{\balpha}{\boldsymbol{\alpha}}
\newcommand{\sing}{\mathop{\mathrm{sing\ supp}}}
\newcommand{\btau}{\boldsymbol{\tau}}
\newcommand{\E}{\mathbb{G}}

\newcommand{\bm}{\mathsf{BM}^+\left(S^{n-1}\right)}
\newcommand{\tp}{\tilde{p}}
\newcommand{\dhat}[1]{\hat{\hat{#1}}}
\newcommand{\tx}{\widetilde{x}}

\def\squareforqed{\hbox{\rlap{$\sqcap$}$\sqcup$}}
\def\qed{\ifmmode\else\unskip\quad\fi\squareforqed}
\def\smartqed{\def\qed{\ifmmode\squareforqed\else{\unskip\nobreak\hfil
\penalty50\hskip1em\null\nobreak\hfil\squareforqed
\parfillskip=0pt\finalhyphendemerits=0\endgraf}\fi}}

\smartqed

\smartqed

\title{The frequency-dependent Wright-Fisher model: diffusive and non-diffusive approximations.}

\author{Fabio A. C. C. Chalub}
\address{Departamento de Matem\'atica 
and Centro de Matem\'atica e Aplica\c c\~oes, 
Universidade Nova de Lisboa,
Quinta da Torre, 2829-516, Caparica, Portugal.}
\email{chalub@fct.unl.pt}

\author{Max O. Souza}
\address{Departamento de Matem\'atica Aplicada,
Universidade Federal Fluminense,
R. M\'ario Santos Braga, s/n, 22240-920, Niter\'oi, RJ, Brasil.}
\email{msouza@mat.uff.br}

\date{\today}

\begin{document}
 
\maketitle

\begin{abstract}

 We study a class of processes that are akin to the Wright-Fisher model,  with transition probabilities  weighted in terms of the frequency-dependent fitness of the population types. 
By considering an approximate weak formulation of the discrete problem, we are able to derive a corresponding continuous weak formulation for the probability density. Therefore, we obtain a family of partial differential equations (PDE) for the evolution of the probability density, and which will be an approximation of the discrete process in the joint large population, small time-steps and weak selection limit.
If the fitness functions are sufficiently regular, we can recast the weak formulation in a more standard formulation, without any boundary conditions, but supplemented by a number of conservation laws. The equations in this family can be purely diffusive, purely hyperbolic or of convection-diffusion type, with frequency dependent convection. The particular  outcome will depend on the assumed scalings. The diffusive equations are of the degenerate type; using a duality approach, we also obtain a frequency dependent version of the Kimura equation without any further assumptions.   We also show that the convective approximation is related to the replicator dynamics and provide some estimate of how accurate is the convective approximation, with respect to the convective-diffusion approximation. In 
particular, we show that the mode, but not the expected value, of the probability distribution is modelled by the replicator dynamics. Some numerical simulations that illustrate the results 
are also presented. 
\end{abstract}

\section{Introduction}
\label{sec:intro}

Evolution is naturally a multiscale phenomenon~\citep{Keller,Metz_mdl}. The choice of right scale to describe a particular problem has as much art as science. For some populations
(e.g, with non overlapping generations) a discrete time provides adequate description; for different examples, this is excessively simplifying. Large populations can be described
as infinite (in order to use differential equations, for example), but this imposes limitations in the time validity of the model~\citep{Chalub_Souza:TPB_2009}. On the other hand, some finite population effects,
like, for example, the bottleneck effect, will be missing in any description relying in infinite populations~\citep{Hartl_Clark}.

In this vein, diffusion approximations, frequently used for large populations and long time scales, enjoy a long tradition in population genetics. This tradition  dates back as early as the work by \cite{Feller1951} and references there in. In particular, diffusion approximations were implicitly used in the pioneering works of \cite{Wright1,Wright2} and \cite{Fisher1,Fisher2}. These efforts have been further developed in a number of directions as, for instance, in the studies on multispecies models in \cite{Sato1976a,Sato1983}; see
 also the review in \cite{Sato1978}. Subsequently, \cite{EthierKurtz} systematically studied the approximation of finite Markov chain models by diffusions. In particular, they showed the validity of a diffusion approximation to a multidimensional Wright-Fisher model, in the regime of weak selection, and linear fitness. This led to a notable progress in diffusion theory, as reported for instance in \citep{EthierKurtz,StrockVarandhan:1997}. This considerable progress, in turn, led to a large use of diffusion theory in population genetics, which can be verified in contemporary introductions to the subject~\citep[see][]{Ewens,EtheridgeLNM}.

There is also a more heuristic approach, called the Kramers-Moyal expansion, where the kernel of the master equation of the stochastic
process is fully expanded in a series. The diffusion approximation 
can be viewed as a Kramers-Moyal expansion truncated at second order. 
Although it is commonly claimed that the  full expansion 
is needed in order to obtain a continuous approximation of discrete processes, 
it is known that under various conditions discrete Markov chains can be approximated by diffusions; cf. \cite{EthierKurtz} and \cite{StrockVarandhan:1997} 
for instance.
In this work, we shall show that under a number of conditions similar results hold for
the discrete processes considered.
See~\citep{VanKampen} for a discussion about this and other techniques for continuous approximations of discrete processes.

As observed above, results along similar lines had been obtained earlier by a number of authors \citep{Feller1951,EthierKurtz,Ewens}.
These works approach the problem mostly within a probabilistic framework, while here we take a pure analytical setting, and this brings two immediate consequences: firstly, we are able to directly derive a weak formulation for the  forward Kolmogorov equation, assuming only continuity of the fitness functions 
in contrast to the weak formulation for the backward  problem for Feller processes; see~\cite{RogersWilliams:2000a,RogersWilliams:2000b}
. The second, and possibly most important consequence, is that we are able to deal with a variety of scalings for the evolution problem. This yields a full family of evolution problems: genetic-drift dominated evolution, which is described by a diffusion equation; selection dominated evolution, which is governed by a hyperbolic equation; and an evolutionary dynamics, where the two forces are balanced, which is governed by a convection-diffusion equation that we term replicator-diffusion.

If we assume some more regularity of the fitness functions, we can then recast the weak formulation in a strong formulation. In this case, we cannot impose  any boundary conditions, but  we must supplement it by a number of conservation laws, namely that the probability of fixation of each type for a given probability density of population, in any time, must be the same as in the initial time. The conservation laws are used to circumvent the impossibility of imposing boundary conditions when the boundaries are absorbing. 

 Furthermore, by a duality argument we obtain the backward equation formulation. For the particular case of linear fitness and  balanced scalings, we then recover the classical result by \cite{EthierKurtz}. Additionally, by an appropriate combination of the weak and strong formulations, we are able to give a  complete description of the forward solution.

A complimentary approach to the study of evolution, based on evolutionary game theory, has also been developed \citep[cf.][]{JMS} with conclusions that are not always compatible with  results from  diffusion theory. 
As an example, diffusion models without mutation lead to the fixation of a homogeneous population, while frequency dependent models associated to the replicator dynamics\footnote{In this work, we will use the expressions ``replicator dynamics'', ``replicator equation'' and ``replicator system'' indistinctly.}
 may lead to stable mixed populations. For an introduction to  evolutionary game theory and replicator dynamics, we refer the reader to \cite{HofbauerSigmund} and \cite{Weilbull}. 
The replicator equation was also modified to introduce stochasticity at population level~\citep{Fudenberg_Harris,Foster_Young}. 
Relations between the matching scheme in a population and the deterministic approximation of its stochastic evolution are studied in~\citep{Boylan_1992,Boylan_1995}.

Consistent interaction among these two modelling schools have been attempted by a number of authors, with 
different degrees of success~\citep[see][]{Traulsen_etal_2006,LessardLadret,Lessard_2005,McKaneWaxman07,Waxman2011,ChampagnatFerriereMeleard_TPB2006,ChampagnatFerriereMeleard_SM2008,Fournier_Meleard_AAP2004,Molzon_2009,BenaimWeibull_2003,CorradiSarin_2000,TraulsenClaussenHauert_PRE2012}. 
We will show, as in many of these works, that both descriptions --- the one based on the diffusion approximation and the one based on the replicator dynamics --- are correct as models for the evolutionary dynamics of a given trait, but in
different scalings. As a by product, we will provide a generalisation of the Kimura equation valid for an arbitrary number of types and general fitness functions.
The long time asymptotics of both descriptions will suggest that the replicator equation is 
a model with limited time validity, given a certain maximum admissible error. Such a limitation is to be naturally anticipated on the grounds that the diffusion process will be eventually absorbed, while the replication dynamics might converge to an equilibrium in the interior.
We also confirm that the solution of the replicator equation indicates the most probable state (mode) of the population, conditional on not have been absorbed. Hence, it does not necessarily indicate the expected value of the 
trait.

The work presented here is a development of earlier work in \citep{Chalub_Souza:TPB_2009,Chalub_Souza:CMS_2009}: the former studying the derivation and convergence of the Moran model with two types to the 1-d version of the replicator-diffusion equation discussed here, and the latter with a comprehensive analytical study of the 1-d replicator-diffusion equation.  The derivation of the continuous model in \cite{Chalub_Souza:TPB_2009} hinged on the idea that a formal expansion of master equation, but with control of the local error,  and results on well-posedness of the continuous classical problem can be brought together via numerical analysis approximation results. This combination then yielded  uniform convergence, in any proper closed sub interval of [0,1], of the rescaled probabilities of the discrete model to the continuous probability density. This convergence result, combined with the analytical results in \cite{Chalub_Souza:CMS_2009} on a weak formulation that satisfies the conservation laws provided 
a continuous measure solution. The discrete process then converges weakly  towards such a solution, on a neighbourhood of each endpoint, but uniformly as described above. To study the Wright-Fisher continuous limits, however, we took a different route. 
This allows us to derive an approximate discrete weak formulation of the discrete process, with global error control. Further, by  embedding the discrete probabilities in an appropriate measure space, we could use compactness arguments to obtain the continuous limit. Thus, in this setting both the weak formulation and the weak convergence of the discrete model to the continuous one follows with considerable less effort, but we do not get the improved convergence on the interior.

\subsection{Scalings, limits and approximations}
\label{ssec:DCA}

In order to be able to study  more general models, we follow the approach used by the authors in \cite{Chalub_Souza:TPB_2009}. In particular, we are interested not only in diffusion approximations, but in approximations that can be consistent with the dynamics of the corresponding discrete process. 

We begin with a definition:
\begin{deff}\label{def:DCA}
  We shall say that a simplified model $\mathcal{M}_0$
is an approximation
of the family of detailed models $\mathcal{M}_\gamma$, $\gamma>0$, in a sense $\chi$, where $\chi$ is
an appropriate metric as, for instance, any norm in a suitable space of functions (e.g.,  $L^1$, $L^2$, $L^\infty$, etc) if the following holds true:
\begin{enumerate}
\item Consider a certain family of initial conditions $h^\ini_\gamma$ such that
$\lim_{\gamma\to0}h^\ini_\gamma=h^\ini_0$, in the sense $\chi$;
\item Evolve through the model $\mathcal{M}_\gamma$ the initial condition $h^\ini_\gamma$
and through the model $\mathcal{M}_0$ the initial condition $h^\ini_0$ until the time
$t<\infty$ obtaining $h_\gamma(t)$ and $h_0(t)$ respectively;
\end{enumerate}
If for all $t<\infty$ we have that $\lim_{\gamma\to0} h_\gamma(t)=h_0(t)$, in the sense $\chi$, then we say that the model $\mathcal{M}_\gamma$ converge in
the limit $\gamma\to0$,  in the sense $\chi$, to the model $\mathcal{M}_0$. If, furthermore, this convergence is uniformly in $t\in[0,\infty)$, then we say that the model $\mathcal{M}_\gamma$ converge in
the limit $\gamma\to0$ to the model $\mathcal{M}_0$ uniformly in time. 
\end{deff}

Some examples of the relation between detailed and simplified models  are listed in Table~\ref{table:drm}.

In general, some extra assumptions are frequently required to allow the passage to the limit. If, for example, there are more than one small parameter in the detailed model, it
is natural to assume a relationship among them, called \textsl{scaling}, as, in general, the limit model will depend on how these
parameters approaches zero. Other assumptions may also be necessary, as it will be discussed in the next paragraph.
 The process of taking the limit of a family of models,
considering a given scaling, will be called ``the thermodynamical limit''; by extension,
we shall also call the limit model the \textsl{thermodynamical limit}.
In this work, depending on the precise choice of the scaling, the limit equation
can be of drift-type (a partial differential equation fully equivalent to the
replicator equation or system), of purely diffusion type, or, in a delicate
balance, of drift-diffusion type.
\begin{table}
{\footnotesize
\begin{center}
 \begin{tabular}{c|c|c}
Detailed model&Meaning of parameter $\gamma$&Simplified model\\
\hline
 Kinetic models&mean free path&hydrodynamical models\\
Othmer-Dumbar-Alt model&mean free path&Keller-Segel model\\
Quantum Mechanics&rescaled Planck constant&Classical Mechanics\\
Relativistic mechanics&(rescaled light velocity)${}^{-1}$&Non-relativistic Mechanics\\
Moran process&inverse of population size&replicator-diffusion equation\\
Moran process&inverse of population size&replicator equation
 \end{tabular}
\end{center}
\caption{Detailed and simplified models. The last two lines state that both the replicator equation and the replicator
diffusion equation approximates the Moran process. References to these works are
\citep{Bardos_Golse_Levermore_I,Bardos_Golse_Levermore_II,Cercignani,Othmer_Hillen_I,Othmer_Hillen_II,CMPS,Stevens_2000,
Hepp,Cirincione,Bjorken_Drell,Chalub_Souza:CMS_2009,Chalub_Souza:TPB_2009}.}\label{table:drm}
}
\end{table}

In what follows, an important and natural assumption that must be introduced in order that we have an approximation in the sense of definition~\ref{def:DCA} is the so-called \textsl{weak selection principle}, to be precisely stated in equation~(\ref{WSP}). Generally speaking, we assume that the \textsl{fitness} of a given
individual 
converges to 1 when the time separation between two successive generations $\Delta t$ approaches zero. This is a
natural assumption when we consider that two successive generations collapses into a
single one. However, in most of the literature, the weak selection principle
is assumed in the limit of $N\to\infty$, where $N$ is the population size. 
Although they are equivalent (as we shall
assume a certain scaling relation between $N$ and $\Delta t$), we consider our approach more natural.

 In this work, we will consider as the detailed model, the Wright-Fisher process, to be studied in detail for finite populations in Section~\ref{sec:finite}: an evolutionary process for an asexual population of $N$ individuals, constant in size, divided in $n$ different types, that evolves according to a specific rule, with fixed time separation between generations of $\Delta t>0$ (the detailed model in the discussion above, where $\gamma$ is the inverse of the population size --- or, as we shall see, equivalently, the inter generation time).

In short, given the weak-selection principle, we are able to find a precise scaling that yields, as the thermodynimical limit. a parabolic equation with degenerate  diffusion at the boundaries: namely the replicator-diffusion equation.  Under a range of different scalings, however,  we shall also obtain a simpler first order differential system: the replicator equation. This simpler model turns out to be  a good approximation for the detailed model over a short time scale, if genetic drift is weak or selection is strong\footnote{Strong selection in this context is
 not directly related or opposed to weak selection as introduced before.}. For the  former approximation, nevertheless, we shall be able to show that $\lim_{\gamma\to 0}\lim_{t\to\infty}h_\gamma(t)=\lim_{t\to\infty}\lim_{\gamma\to 0}h_\gamma(t)$, and hence  we conjecture that such an approximation should be uniform in time.

\subsection{Outline}

Section~\ref{sec:prelim} introduces the basic notation and provides an extended abstract of our main results. In Section~\ref{sec:finite}, we  review some classical results about the discrete process (the finite population Wright-Fisher process); we also show the existence of a number of associated conservation laws,
and explicitly obtain the first moments of the Wright-Fisher process.
 In Section~\ref{sec:infinite}, with the assumption of weak-selection,
we obtain a family of continuous limits of the Wright-Fisher process depending on the scalings that are derived within a weak formulation, with solutions in appropriate measure spaces. In particular, we derive the replicator-diffusion equation, and show that it satisfies continuous counterparts of the conservation laws for the discrete process. We then continue the study of the replicator-diffusion equation in Section~\ref{sec:forward}, where we derive the main properties of its solutions, including a description of the solution structure as a regular part and a sum of singular measures over the sumbsimplices, and the large time convergence to a sum of Dirac measures over the vertices of the simplex.
  We also show that the probability distribution associated with all types in the
population concentrates along the evolutionary stable states.
Additionally, in Subsection~\ref{ssec:duality},  we obtain the backward equation as the proper dual of the replicator-diffusion equation,
providing a consistent generalisation of the Kimura equation for the $n$ types and arbitrary fitness functions.
 In Section~\ref{sec:replicator}, we study the replicator equation and show that, in the regime of strong selection the solutions to the replicator-diffusion will be well approximated by the
  solutions to the replicator equation within a finite time interval. 
Numerical examples are given in Section~\ref{sec:numerics}, where we also point out that, for intermediate times and large but finite populations, the replicator equation will approximate the mode of the discrete  evolution, but not the expected value of a given trait. Conclusions are presented in Section~\ref{sec:conclusion}.

\section{Preliminaries and main results}
\label{sec:prelim}

We begin by introducing the space of states for the evolution:
\begin{deff}\label{def:simplex}
Let $\R_+=[0,\infty)$. We define the $n-1$ dimensional simplex
\[
S^{n-1}\bydef\left\{\mathbf{x}\in\R_+^n\left|\, |\mathbf{x}|\bydef\sum_{i=1}^nx_i=1\right.\right\}.                                   
\]
We also define the set of vertices of the simplex
$\Delta S^{n-1}\bydef\{\mathbf{x}\in S^{n-1}|\exists i, x_i=1\}$,
its interior $\mathrm{int}\,{S^{n-1}}\bydef\{\mathbf{x}\in S^{n-1}|\forall i, x_i>0\}$
and its boundary $\partial S^{n-1}=S^{n-1}\backslash\mathrm{int}\, S^{n-1}$.
 The \textsl{state} of the population is a vector $\mathbf{x}\in S^{n-1}$.
The elements of $\Delta S^{n-1}$ are denoted $\mathbf{e}_i$, $i=1,\dots,n$ and called
``homogeneous states''. A vector $\mathbf{x}\in S^{n-1}\backslash\Delta S^{n-1}$ is a ``mixed state''. 
\end{deff}

In what follows, we let $p(\mathbf{x},t)$ to be the probability density of finding the population at state $\mathbf{x}\in S^{n-1}$ at time $t\ge 0$. 

\begin{deff}
 The fitness of type $i$, $i=1,\dots,n$ is a continuous function 
$\psi^{(i)}:S^{n-1}\to\R$, and the average fitness in a given population 
is given by  $\bar\psi(\mathbf{x})\bydef\sum_{i=1}^nx_i\psi^{(i)}(\mathbf{x})$. Note that we consider
the fitness function to not depend explicitly on time. 
\end{deff}

In this work, we derive a family of detailed models described by a parabolic equation of 
drift-diffusion type, with degenerated coefficients~\citep{DiBenedetto93,CS76},
defined in the simplex $S^{n-1}$, called \textsl{the replicator-diffusion equation}, namely:
\begin{equation}\label{replicator_diffusion_eps}
\left\{
\begin{array}{l}
\partial_t p=\mathcal{L}_{n-1,x}p\bydef\frac{\kappa}{2}\sum_{i,j=1}^{n-1}\partial_{ij}^2\left(D_{ij}p\right)
-\sum_{i=1}^{n-1}\partial_{i}\left(\Omega_ip\right)\ ,\\
D_{ij}\bydef x_i\delta_{ij}-x_ix_j\ ,\\
\Omega_i\bydef x_i\left(\psi^{(i)}(\mathbf{x})-\bar\psi(\mathbf{x})\right)\ ,
\end{array}\right.
\end{equation}
with $i,j=1,\dots,n-1$, $\kappa> 0$, and
where $\delta_{ij}=1$ if $i=j$ and 0 otherwise is the Kronecker delta. The above equation has a solution in the classical sense (i.e., everywhere differentiable).
Furthermore, in the  classical sense, it is a well posed problem, without any boundary conditions. However, this classical solution is not
the correct limit of the discrete process. In order to find the correct limit, 
equation (\ref{replicator_diffusion_eps}) is to be supplemented with
$n$ conservation laws. From now on, whenever we refer to the replicator-diffusion equation~(\ref{replicator_diffusion_eps}), we are implicitly assuming
these conservation laws.

Our main conclusions are:
\begin{enumerate}
\item  An analysis of the equation (\ref{replicator_diffusion_eps}) leads to a unique solution of measure type. This will require definitions of appropriate functional spaces.
\item  This unique solution approximates, in the thermodynamical limit, the evolution of a discrete population by the Wright-Fisher process 
pointwise for any time. In addition, the large time asymptotics is consistent with the discrete model.
\item  A reduced model, obtained by setting $\kappa=0$ in (\ref{replicator_diffusion_eps}) (with only one conservation law), is
shown to be equivalent to the replicator dynamics. This will suggest that the replicator dynamics approximates the discrete process for any $t$, however with an error increasing in $t$ along a fixed discretisation
\item Furthermore, the
solution of the replicator equation models the time evolution of the mode of the probability distribution associated to the discrete process (and not the \textsl{expected value}
of the same distribution);  
\item A frequency dependent generalisation of the Kimura equation for an arbitrary number of types is
obtained by looking at the dual problem for (\ref{replicator_diffusion_eps}).
\end{enumerate}

Before going into the technical details, we explain the last paragraph a little further.

Equation~(\ref{replicator_diffusion_eps}) has two natural time scales, one for the
natural selection (the mathematical drift and, as we shall see, fully compatible with
the replicator equation), the second for the genetic drift
(the mathematical diffusion). That is why we call equation~(\ref{replicator_diffusion_eps})
together with the conservation laws to be introduced in Subsection~\ref{ssec:conservation},
the ``replicator-diffusion equation''.
More precisely, the solution of the replicator-diffusion
equation when $\kappa=0$ (which is of hyperbolic type) is the leading order term
of the solution $p_\kappa$ of the replicator-diffusion equation for small $\kappa$ (i.e., large fitness
and/or short times). The replicator-diffusion equation with zero diffusion ($\kappa=0$) happens
to be the replicator equation (or system)~\citep{HofbauerSigmund}.
In an appropriate sense, to be made precise in Section~\ref{ssec:local} (Theorem~\ref{thm:conv_to_replicator}), we have
$p_\kappa\stackrel{\kappa\to 0}{\longrightarrow} p_0$,  pointwise, but not uniformly in time.

Theorem~\ref{thm:conv_to_replicator} cannot be made uniform in time, for general fitness functions and initial conditions,
as the Wright-Fisher process always converge in $t\to\infty$
to a linear combination of homogeneous states, while it is possible that the solution of the
replicator equation converges to a stable mixed state. 

The former statement is the mathematical formulation of a known principle in evolutionary biology that states that ``given enough time every
mutant gene will be fixed or extinct.''~\citep{Kimura}. This means that the final state of the replicator-diffusion equation
with any $\kappa>0$ will be a linear combination of Dirac deltas at the vertices of the simplex $S^{n-1}$. Actually, for any 
positive time, the solution of equation~(\ref{replicator_diffusion_eps}) with the conservation laws described above is a sum of
 a classical function in the simplex plus a sum  of singular measures over all the subsimplexes on $\partial S^{n-1}$ and, inductively,
 on their boundaries subsimplexes. In particular, we shall have also Dirac measures supported on the vertices of the simplex. 
These measures appears immediately, i.e., for any
$t>0$. This represents the fact that in a single step there is a non zero --- however, small --- probability that the population reaches
 fixation through Wright-Fisher evolution. The full evolution and the final states of the replicator-diffusion equation will
be studied in Section~\ref{sec:forward}.

From the practical point of view, we are, however, often interested in transient states (``in the long run, we are
all dead'', said John Maynard Keynes), specially because the transient states become more and more
important for the discrete evolution as the population size increases. Heuristically, when the
population is large the stochastic fluctuations
decrease in importance, and therefore, its evolution is deterministic. The associated limit will be given by
equation~(\ref{replicator_diffusion_eps}), with $\kappa=0$, i.e., the hyperbolic limit of equation~(\ref{replicator_diffusion_eps}).  
This equation does not develop finite-time singularities. 

The relationships between the three models is summarised in Figure~\ref{fig:relations}.

\begin{figure}
\centering
\includegraphics[width=1\linewidth]{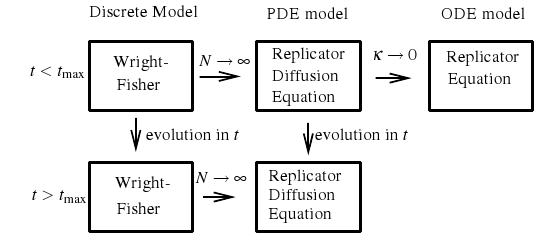}
\caption{The boxes in the figure represents the solutions of three different models: the Wright-Fisher
process (finite population $N$), the replicator-diffusion equation (positive diffusion $\kappa$)
and the replicator equation. The vertical axis indicates the arrow of time (top-down), and the horizontal
axis indicates, first the large population limit, secondly the no-diffusion limit. 
Consider that there is 
 a maximum acceptable error $\epsilon$ (in the $L^\infty$ norm) between the Wright-Fisher model (suitably \emph{Radonmised} --- see Subsection~\ref{ssec:continuos_rep}) and 
 the continuous approximation. Therefore, there is a population size $N_0$ such that for $N>N_0$ the difference between the
 replicator-diffusion equation and the discrete model is less than $\epsilon$. For the replicator approximation, for any $N$, it may exist maximal time $t_{\max}(N)<\infty$ such that for $t>t_{\max}$
  the error is too large.
}
\label{fig:relations}
\end{figure}

Finally, we observe that the natural formulation for the continuous limit is the weak one. For such a formulation, we only require the fitness functions to be continuous. If, in addition, these functions are also Lipschitz, we can then recast the problem in a strong sense, provided that we supplement it with the conservation laws. Finally, requiring the fitness functions to be smooth allows for a number of results about the  solutions to be easily derived. In particular, one can prove a structure theorem that show that the problem is equivalent to a hierarchy of classical degenerate problems, provided that some members are interpreted as densities for singular measures.

\section{The discrete model}
\label{sec:finite}

In this section, we study the discrete model, i.e., the Wright-Fisher model for constant population, arbitrary number
of types and arbitrary fitnesses functions. We start, in Subsection~\ref{ssec:prelim} with basic definitions; in 
Subsection~\ref{ssec:sstates_fstates_claws} we briefly review some important results in the literature.
We also prove that the discrete process has as many linear conservation laws as types. Additionally, we also show that the
final state is a linear superposition of these independent stationary states, with coefficients that depend on the initial condition
and that can be calculated from a set of $n$ linearly independent conservation laws. All these results will be useful in
the correct determination of the continuous process, to be done in Sections~\ref{sec:infinite} and~\ref{sec:forward}.
 The discrete Wright-Fisher process was studied, with different level of details in, for
 example,~\citep{Ewens,Nowak_2006,Imhof_Nowak_2006}, but, to the best of our knowledge the conservation laws associated to the process
 were overlooked. 

The fact that the final state in the Wright-Fisher process, among others, for a finite population is always homogeneous was also
a matter of dispute with respect to the validity of the modelling~\citep{Vickery_1988,Smith_1988}. As we will shortly see in this work, this dispute is basically a consequence of the existence of two different time scales hidden in the model: the non-diffusive (drift) and the diffusive one. See also~\cite{EthierKurtz} and \cite{EtheridgeLNM} and references therein for a discussion in the r\^ole of time scales.

\subsection{Preliminaries}
\label{ssec:prelim}

We consider a
fixed size population of $N$ individuals
at time $t$ consisting of a fraction $x_i\in\left\{0,\frac{1}{N},\frac{2}{N},\dots,1\right\}$ of
individuals of type $i=1,2,\cdots,n$. 
The population evolves in discrete generations with time-step separation of $\Delta t$.
We introduce the following notation:
\begin{deff}
The \textsl{state} of a population is defined by a vector in
the $N$-discrete $n-1$-dimensional simplex 
\[
S^{n-1}_N\bydef\left\{\mathbf{x}=(x_1,\cdots,x_n)\big||\mathbf{x}|\bydef\sum_{i=1}^nx_i=1, 
x_i\in\left\{0,\frac{1}{N},\frac{2}{N},\cdots,1\right\}\right\}\ .
\]
We also define the set of vertices of the $n-1$-dimensional simplex
\[
\Delta S^{n-1}_N\bydef\{\mathbf{x}\in S^{n-1}_N|\exists i, x_i=1\}=\left\{\mathbf{e}_i|i=1,\dots,n\right\}\ .
\]
The elements of $\Delta S^{n-1}_N$ are called \textsl{the homogeneous states}.
To each type we attribute a
function, called fitness, $\Psi^{(i)}_{\Delta t}:S^{n-1}_N\to(0,\infty)$.
It is convenient to assume that $\Psi^{(i)}_{\Delta t}$ is a discretisation
of a continuous function on the simplex $S^{n-1}$; more assumptions on $\Psi^{(i)}_{\Delta t}$ will be introduced in
Section~\ref{sec:infinite}.
\end{deff}

A population at time $t+\Delta t$ is obtained from 
the population at time $t$ choosing $N$ individuals with
probability proportional to the fitness.
More precisely, we define the average fitness $\bar\Psi_{\Delta t}(\mathbf{x})=\sum_{i=1}^nx_i\Psi^{(i)}_{\Delta t}(\mathbf{x})$ and
then the transition probability from a population 
at state $\mathbf{y}$ to a population at state $\mathbf{x}$ is given by
\begin{equation}\label{tran_prob_WF}
 \Theta_{N,\Delta t}(\mathbf{y}\to\mathbf{x})=\frac{N!}{(Nx_1)!(Nx_2)!\cdots (Nx_n)!}
\prod_{i=1}^n\left(\frac{y_i\Psi^{(i)}_{\Delta t}(\mathbf{y})}{\bar\Psi_{\Delta t}(\mathbf{y})}\right)^{Nx_i}\ .
\end{equation}

The evolutionary process given by a Markov chain with transition probabilities given by
equation~(\ref{tran_prob_WF}) is called \textsl{the (frequency dependent) Wright-Fisher process}.

Let $\mathcal{P}(t)=\left(P(\mathbf{x},t)\right)_{\mathbf{x}\in S^{n-1}_N}$, with
\[
\mathcal{P}\in\left\{P:S^{n-1}_N\times\R_+\to\R_+|\sum_{\mathbf{x}\in S^{n-1}_N}P(\mathbf{x},\cdot)=1\right\}\ ,
\]
where $P(\bx,t)$
is the probability of finding the population at a given 
state $\mathbf{x}\in S^{n-1}_N$  at time $t$.
Then, the evolution is given by the so called ``master equation'':
\begin{equation}\label{discrete_evolution}
P(\mathbf{x},t+\Delta t)=
(\mathcal{T}\mathcal{P}(t))(\mathbf{x})\bydef
\sum_{\mathbf{y}\in S^{n-1}_N}\Theta_{N,\Delta t}(\mathbf{y}\to\mathbf{x})P(\mathbf{y},t)\ .
\end{equation}

\subsection{Stationary states, final states and conservation laws}
\label{ssec:sstates_fstates_claws}

We call an homogeneous population a population of a single type, i.e., $P(\mathbf{x},t)=\hat P_{\mathbf{v}}(\mathbf{x})$
for $\mathbf{v}\in\Delta S^{n-1}_N$,
where 
\[
 \hat P_{\mathbf{x}}(\mathbf{y})=\left\{
\begin{array}{ll}
 1\ ,&\qquad \mathbf{y}=\mathbf{x}\ ,\\
0\ ,&\qquad\mathbf{y}\ne\mathbf{x}\ .
\end{array}
\right.
\]
From the inner product definition:
 \[
\langle v,w\rangle\bydef\sum_{\mathbf{x}\in S^{n-1}_N}v(\mathbf{x})w(\mathbf{x})\ ,
\]
it follows immediately that 
$\langle\hat P_{\mathbf{x}},\hat P_{\mathbf{y}}\rangle=\delta_{\mathbf{x},\mathbf{y}}=1$ if $\mathbf{x}=\mathbf{y}$ and 0 otherwise.

Now, we state classical results for the Wright-Fisher process that will be useful in the sequel. Firstly, from the Perron-Frobenius theorem, the operator $\mathcal{T}^\infty:\slf^{n-1}_N\to\slf^{n-1}_N$, $\mathcal{T}^\infty P\bydef\lim_{m\to\infty}\mathcal{T}^m P$ is
well defined. For details in the lemma below, the interested reader should consult~\cite{Karlin_Taylor_first}.

\begin{lem}\label{lem:classical}
A function $f$ defined in $S^{n-1}_N$ is a fixed state of the operator $\mathcal{T}$
if, and only if, $f$ is a linear combination of homogeneous states. In particular, $\mathcal{T}$ has exactly $n$ linearly independent eigenfunctions associated to the eigenvalue $\lambda=1$.
For all non-negative initial condition $P^\ini$, the final result is
a linear combination of homogeneous states,
\[
\mathcal{T}^\infty P\bydef\lim_{t\to\infty}P(\cdot,t)=
\sum_{i=1}^n F^{(i)}_{P_{\mathbf{e}_i}}\hat P_{\mathbf{e}_i}\ ,
\]
\end{lem}
where $F^{(i)}_{P}\bydef\lim_{m\to\infty}\langle\mathbf{e}_i,\mathcal{T}^m P\rangle$ is the fixation probability of the
type $i$ in a population initially with a probability distribution $P\in\slf^{n-1}_N$. 

\begin{deff}
 We define a \textit{linear conservation law} as one given by a linear functional $\mathsf{L}$ over the functions of $S^{n-1}_N$
such that $\mathsf{L}\left(\mathcal{P}(t+\Delta t)\right)=\mathsf{L}\left(\mathcal{P}(t)\right)$. A set of linear conservation laws is linearly independent,
if the only linear combination providing a trivial conservation law $\mathsf{L}(\mathcal{P}(t))=0$ is the trivial one.
\end{deff}

\begin{prop}
 Define  $\mathsf{F}^{(i)}\bydef\sum_{\mathbf{x}\in S^{n-1}_N}F^{(i)}_{\hat P_{\mathbf{x}}}\hat P_{\mathbf{x}}$, $i=1,\dots, n$, a functional over $S^{n-1}_N$. Therefore $\mathsf{F}^{(i)}(\mathbf{x})$ is 
the fixation probability of the type $i$ associated to the initial condition $\mathbf{x}$. 
Finally, the 
set  $\{\mathsf{F}^{(1)},\dots,\mathsf{F}^{(n)}\}$ is a basis for the set of linear conservation laws associated to the operator $\mathcal{T}$. 
\end{prop}

\begin{proof}
From the fact that
\[
 \mathcal{T}^\infty P=\sum_{i=1}^n F^{(i)}_P\hat P_{\mathbf{e}_i}
\]
we find
\begin{equation*}
 F^{(i)}_{P}=\left(\mathcal{T}^\infty P\right)\left(\mathbf{e}_i\right)=\langle\mathcal{T}^\infty P,\hat P_{\mathbf{e}_i}\rangle
=\langle P,\left(\mathcal{T}^\dagger\right)^\infty\hat P_{\mathbf{e}_i}\rangle\ .
\end{equation*}
In particular 
\[
 \sum_{\mathbf{x}\in S^{n-1}_N}F^{(i)}_{\hat P_{\mathbf{x}}}\hat P_{\mathbf{x}}=\sum_{\mathbf{x}\in S^{n-1}_N}
\langle\hat P_{\mathbf{x}},\left(\mathcal{T}^\dagger\right)^\infty\hat P_{\mathbf{e}_i}\rangle\hat P_{\mathbf{x}}.
\]
Finally,
\[
 \left(\mathcal{T}^\dagger\right)^\infty\hat P_{\mathbf{e}_i}=\sum_{\mathbf{x}\in S^{n-1}_N}F^{(i)}_{\hat P_{\mathbf{x}}}\hat P_{\mathbf{x}}=\mathsf{F}^{(i)}\ .
\]
Therefore, $\mathsf{F}^{(i)}$ is an eigenvector of $\mathcal{T}^\dagger$.
In particular,
\[
 \mathsf{F}^{(i)}(\mathbf{e}_j)=\langle\left(\mathcal{T}^\dagger\right)^{\infty}\hat P_{\mathbf{e}_i},\hat P_{\mathbf{e}_j}\rangle
=\langle\hat P_{\mathbf{e}_i},\mathcal{T}^\infty\hat P_{\mathbf{e}_j}\rangle=
\langle\hat P_{\mathbf{e}_i},\hat P_{\mathbf{e}_j}\rangle=\delta_{ij}\ .
\]
It is immediate to prove that they are linearly independent; let $\alpha_1,\dots,\alpha_n$ such
that $\sum_{i=1}^{n}\alpha_i\mathsf{F}^{(i)}=0$, i.e., 
for every $\mathbf{x}\in S^{n-1}_N$,
$
\sum_{i=1}^{n}\alpha_i\mathsf{F}^{(i)}(\mathbf{x})=0
$.
Using $\mathbf{x}=\mathbf{e}_i$, we conclude that $\alpha_i=0$, and then $\{\mathsf{F}^{(1)},\dots,\mathsf{F}^{(n)}\}$ is a basis for the eigenspace of $\mathcal{T}^\dagger$ associated to
$\lambda=1$.

 Now, consider a linear conservation law $\mathsf{L}$. From standard representation theorems, there is a vector $w\in S^{n-1}_N$ such that
 \begin{equation}\label{discrete_conservation_laws}
  \langle\mathcal{P}(t),w\rangle=\mathsf{L}(\mathcal{P}(t))=\mathsf{L}(\mathcal{P}(t+\Delta t))
=\langle\mathcal{T}\mathcal{P}(t),w\rangle=\langle\mathcal{P}(t),\mathcal{T}^\dagger w\rangle\ .
 \end{equation}
Therefore, $w$ is an eigenvector of $\mathcal{T}^\dagger$ associated to $\lambda=1$ and then it is a linear combination of $\mathsf{F}^{(i)}$, $i=1,\dots,n$.

\end{proof}

\begin{rmk}
The conservation of probability (the most natural conservation law), follows directly from the equation
\[
 \sum_{i=1}^n\mathsf{F}^{(i)}(\mathbf{x})=\sum_{i=1}^nF^{(i)}_{\hat P_{\mathbf{x}}}=1\ ,\ \forall\mathbf{x}\in S^{n-1}_N\ .
\]
\end{rmk}

\subsection{Properties of the transition kernel}

The probability conservation is a consequence of the definition~(\ref{tran_prob_WF}) and reads
\begin{equation}\label{conservation_of_probability}
\sum_{\mathbf{x}\in S^{n-1}_N}\Theta_{N,\Delta t}(\mathbf{y}\to\mathbf{x})=1, \qquad\forall\mathbf{y}\in S^{n-1}_N\ .
\end{equation}
It also follows from the definition that 
\begin{equation}\label{transition_stationary}
\Theta_{N,\Delta t}(\mathbf{y}\to\mathbf{x})=\left\{
\begin{array}{ll}
1&\quad\text{if}\quad\mathbf{x}=\mathbf{y}\in\Delta S^{n-1}_N\ ,\\
0&\quad\text{if}\quad\mathbf{x}\ne\mathbf{y}\in\Delta S^{n-1}_N\ ,
\end{array}
\right.
\end{equation}
which can be readily interpreted as the absence of mutations in the model.

It will be also convenient to write
\[
 S^{n-1}_{N,\mathbf{x}^\pm}=\left\{\mathbf{y}\in\mathbb{R}^{n-1} | \mathbf{x}\pm\mathbf{y} \in S^{n-1}_N\right\}.
\]
and to introduce
\[
z\tau_i=y_i,\quad z=\frac{1}{\sqrt{N}}.
\quad\text{and}\quad
\mathcal{S}_{\mathbf{x},z}=\{\boldsymbol{\tau}\in\R^n\big|\sum_{i=1}^n\tau_i=0\ \text{and} |\tau_i|<x_i/z\}.
\]

\begin{lem}\label{theta_moments}
Define 
\[
\widetilde{x}_i=\frac{x_i\Psi^{(i)}_{\Delta t}(\bx)}{\bar\Psi(\bx)}
\]
and 
\[
 \E_\Theta[h]=\sum_{z\boldsymbol{\tau}\in S^{n-1}_{N,\mathbf{x}^+}}\Theta(\bx\to\bx+z\boldsymbol{\tau})h(\boldsymbol{\tau})\ ,
\]
where $h:\mathcal{S}_{\mathbf{x},z}\to\R$.

For any $N$, we have
\begin{align*}
\E_\Theta[1]&=1\\
z\E_\Theta[\tau_i]&=\tx_i-x_i\\
z^2\E_\Theta[\tau_i\tau_j]&=(\tx_i-x_i)(\tx_j-x_j)+z^2\left(\delta_{ij}\tx_i-\tx_i\tx_j\right)\\
z^3\E_\Theta[\tau_i\tau_j\tau_k]&=(\tx_i-x_i)(\tx_j-x_j)(\tx_k-x_k)\\
&\qquad + z^2\left[(\delta_{ij}\tx_i-\tx_i\tx_j)(\tx_k-x_k)+(\delta_{ik}\tx_i-\tx_i\tx_k)(\tx_j-x_j) + (\delta_{jk}\tx_j-\tx_j\tx_k)(\tx_i-x_i)\right]\\
&\qquad +
z^4\left[2\tx_i\tx_j\tx_k-(\delta_{ij}\tx_j\tx_k+\delta_{ik}\tx_i\tx_j+\delta_{kj}\tx_i\tx_j)+
\delta_{ij}\delta_{ik}\delta_{jk}\tx_i\right]
 \end{align*}
\end{lem}

\begin{proof}
Let $\mathbf{q}\in\slf^{n-1}$, $\balpha\in N S^{n-1}_N$ and consider the multinomial distribution given by
\begin{equation*}
 f(\mathbf{q},\boldsymbol{\alpha},N)=\frac{N!}{\alpha_1!\cdots \alpha_n!}\prod_{k=1}^nq_k^{\alpha_k},\quad \boldsymbol{\alpha}=(\alpha_1,\ldots,\alpha_n),\quad\sum_i\alpha_i=N\ .
\label{mlt_pmf}
\end{equation*}
Then, $\balpha$ is a vector of random variables with first moments given by
\begin{align*}
& \mathbb{E}[1]=1\ ,\\
& \mathbb{E}[\alpha_i]=Nq_i\ ,\\
& \mathbb{E}[(\alpha_i-Nq_i)(\alpha_j-Nq_j)]=\mathop{\mathrm{Cov}}(\alpha_i,\alpha_j)=N(\delta_{ij}q_i-q_iq_j)\ ,\\
& \mathbb{E}[(\alpha_i-Nq_i)(\alpha_j-Nq_j)(\alpha_k-Nq_k)]= N\left[q_i\delta_{ij}\delta_{kj}-\left(q_iq_k\delta_{ij}+q_iq_j\delta_{kj}+q_kq_j\delta_{ik}\right)+2q_iq_jq_k\right],
\end{align*}
where $\mathbb{E}[\cdot]$ is the expected value under the multinomial distribution.
See~\cite{Karlin_Taylor_intro} for the mean and covariance;  for the sake of completeness, we provide a derivation of the third moment in Appendix~\ref{ap:thirdmomentum}.

Now, note that $\Theta_{N,\Delta t}(\bx,\bx+z\boldsymbol{\tau})=f(\widetilde{\bx},N(\bx+z\boldsymbol{\tau}),N)$.
Therefore, $\balpha=N(\bx+z\btau)$ is a random vector that is multinomially distributed, and  upon substituting $\alpha$ in the multinomial moments --- with $\mathbf{q}=\tilde{\bx}$ --- all the identities follow after some manipulation.
\qed
\end{proof}

\section{Continuations of the discrete  model}
\label{sec:infinite}

The aim of this section is to obtain a differential equation that approximates the discrete evolution,
when the population is large ($N\to\infty$) and there is no time-separation between successive
generations ($\Delta t\to 0$). The relevant variables, $\mathbf{x}\in S^{n-1}$ and $t>0$ will
be forced to be continuous. 

The first four subsections will be devoted to the
development of three models, based on partial differential equations obtained from the Wright-Fisher process, when
$N\to\infty$ and $\Delta t\to 0$ (see equations~(\ref{convective_approximation}), (\ref{diffusive_approximation})
and (\ref{replicator_diffusion}), respectively). There is no ``right choice'' of the simplified model. As we
could expect, simpler models will have a restricted application. For example,
the model given by equation (\ref{convective_approximation}) is equivalent to a system of a ordinary differential equations; 
actually, it is exactly equivalent to the well-know replicator dynamics~\citep[see][]{HofbauerSigmund}.
On the other hand, the diffusive approximation, given by equation~(\ref{diffusive_approximation}), 
is a parabolic partial differential equation that is much simpler to
solve than the full model; in fact, explicit solutions are know using Gegenbauer polynomials~\citep{Ewens}.
Our focus will be on the replicator-diffusion approximation, equation~(\ref{replicator_diffusion}), which we 
expect to be valid uniformly in time.

Results known for the Wright-Fisher process, and
stated in Section~\ref{sec:finite} will guide the derivation, i.e., the choice of the
right thermodynamical limit. We start in Subsection~\ref{ssec:prelimcont} by the asymptotic expansion of the 
transition kernel in the negligible parameters (suitable combinations of $N$ and $\Delta t$); we plug this expansion into
the master equation~(\ref{discrete_evolution}) in Subsection~\ref{ssec:weak-discrete}. In Subsection~\ref{ssec:continuos_rep}, we construct
the continuous version of the discrete probability densities; in particular, we interpolate discrete probabilities in order to represent them by continuous
probability measures; these measures will be central when we finally pass to the limits in Subsection~\ref{ssec:passage_to_the_limit}, obtaining the various continuous
approximations of the discrete model. Finally, in Subsection~\ref{ssec:conservation}, we show that,
for every conservation law of the discrete process,  
there exists a corresponding  conservation law in the continuous model. 
As a by product, the final state of the continuous model shall be a linear superposition of homogeneous states
 (see Lemma~\ref{lem:classical} and compare it with Theorem~\ref{thm:final_state}).

\subsection{Preliminaries}
\label{ssec:prelimcont}

From a biological point of view, the most important assumption in this derivation is the so called weak selection principle
\begin{equation}\label{WSP}
\Psi^{(i)}_{\Delta t}(\mathbf{y})=1+\left(\Delta t\right)^\nu\psi^{(i)}(\mathbf{y}) ,
\end{equation}
where $\psi^{(i)}:S^{n-1}\to\R$ is a continuous function, and $\nu>0$ is a parameter yet to be specified.
In this case, we also have
\[
\bar\Psi_{\Delta t}(\mathbf{y})=1+\left(\Delta t\right)^\nu\bar\psi(\mathbf{y})\ ,
\]
where $\bar\psi(\bx)\bydef\sum_{i=1}^nx_i\psi^{(i)}(\bx)$.

As an immediate corollarium of Lemma~\ref{theta_moments} we have
\begin{cor}\label{cor:theta_moments_asymp}
 Assume the weak-selection principle given by equation~(\ref{WSP}). Then, equations for $\E_\Theta$ in Lemma~\ref{theta_moments} are given by
\begin{align*}
z\E_\Theta[\tau_i]&=x_i\left(\Delta t\right)^\nu\left(\psi^{(i)}(\bx)-\bar\psi(\bx)\right)+\bigO\left(\left(\Delta t\right)^{2\nu}\right)\ ,\\
z^2\E_\Theta[\tau_i\tau_j]&=
\frac{1}{N}\left(x_i\delta_{ij}-x_ix_j\right)+\bigO\left(N^{-1}\left(\Delta t\right)^\nu,\left(\Delta t\right)^{2\nu}\right)\\
z^3\E_\Theta[\tau_i\tau_j\tau_k]&=\bigO\left(\left(\Delta t\right)^{3\nu},N^{-1}\left(\Delta t\right)^\nu,N^{-2}\right)\ .
\end{align*}
\end{cor}

\begin{proof}
 All equations follow from the fact that
\[
 \tx_i=x_i\frac{1+\left(\Delta t\right)^\nu\psi^{(i)}(\bx)}{1+\left(\Delta t\right)^\nu\bar\psi(\bx)}=
x_i\left[1+\left(\Delta t\right)^\nu\left(\psi^{(i)}(\bx)-\bar\psi(\bx)\right)+\bigO\left(\left(\Delta t\right)^{2\nu}\right)\right]
\]
\end{proof}
and from Lemma~\ref{theta_moments}.

\subsection{An asymptotic  weak-discrete formulation}
\label{ssec:weak-discrete}

We rewrite the master equation~\eqref{discrete_evolution} using displacements:
\begin{equation}\label{master_new}
 p(\mathbf{x},t+\Delta t,N)=\sum_{\mathbf{y} \in S^{n-1}_{N,\bx^-}}\Theta_{N,\Delta t}(\mathbf{x}-\mathbf{y}\to\mathbf{x})p(\mathbf{x}-\mathbf{y},t,N).
\end{equation}

We  use our information on moments of the process as follows:
\begin{prop}
 Let $g\in C^{3,1}(\Upsilon)$, where $\Upsilon$ is any open set such that $S^{n-1}\subset\Upsilon$, and consider its restriction to the simplex $\slf^{n-1}$.
Then, we have
\begin{align}\label{eq:weconcludethat}
&  \sum_{\mathbf{x}\in S^{n-1}_N}\left(p(\mathbf{x},t+\Delta t,N)-p(\mathbf{x},t,N)\right)g(\mathbf{x},t)\\
\nonumber
&\quad=
\sum_{\mathbf{x}\in S^{n-1}_N}p(\mathbf{x},t,N)\left[\frac{1}{2N}\sum_{i,j=1}^{n-1}x_i(\delta_{ij}-x_j)\partial^2_{ij}g(\mathbf{x},t)
 +\left(\Delta t\right)^\nu\sum_{j=1}^{n-1}x_j\partial_{x_j}g(\mathbf{x},t)(\psi^{(j)}(\mathbf{x})-\bar\psi(\mathbf{x}))\right]\\
\nonumber
&\qquad\qquad+\bigO\left(N^{-2},\left(\Delta t\right)^{2\nu},N^{-1}\left(\Delta t\right)^\nu \right).
\end{align}
\end{prop}

\begin{proof}
On multiplying equation~(\ref{master_new}) by $g(\mathbf{x},t)$  and summing over $S^{n-1}_N$, we find that:
\begin{align*}
& \sum_{\mathbf{x}\in S^{n-1}_N}p(\mathbf{x},t+\Delta t,N)g(\mathbf{x},t)= 
\sum_{\mathbf{x}\in S^{n-1}_N}\sum_{\mathbf{y}\in S^{n-1}_{N,\mathbf{x}^-}}\Theta_{N,\Delta t}(\mathbf{x}-\mathbf{y}\to\mathbf{x})p(\mathbf{x}-\mathbf{y},t,N)g(\mathbf{x},t)\\
&\quad= \sum_{\mathbf{x}\in S^{n-1}_N}\sum_{\mathbf{y}\in S^{n-1}_{N,\mathbf{x}^+}}\Theta_{N,\Delta t}(\mathbf{x}\to\mathbf{x}+\mathbf{y})p(\mathbf{x},t,N)g(\mathbf{x}+\mathbf{y},t)\\
&\quad=\sum_{\mathbf{x}\in S^{n-1}_N}p(\mathbf{x},t,N)\sum_{z\boldsymbol{\tau}\in S^{n-1}_{N,\mathbf{x}^+}}\Theta(\bx\to\bx +z\boldsymbol{\tau})g(\mathbf{x}+z\boldsymbol{\tau},t)\\
&\quad=\sum_{\mathbf{x}\in S^{n-1}_N}p(\mathbf{x},t,N)\sum_{z\boldsymbol{\tau}\in S^{n-1}_{N,\mathbf{x}^+}}\Theta(\bx\to\bx +z\boldsymbol{\tau})
\left[ g(\mathbf{x},t)+z\sum_{j=1}^{n-1}\tau_j\partial_{x_j}g(\mathbf{x},t)+
\frac{z^2}{2}\sum_{k,l=1}^{n-1}\tau_k\tau_l\partial_{x_kx_l}^2g(\mathbf{x},t)+z^3R(\mathbf{x},\boldsymbol{\tau},t,z)
\right]. 
\end{align*}
where there is a constant $C$, depending only on $g$, such that
\[
|R(\bx,\btau,t,z)|\leq C\|\tau\|^3.
\]

Using Corollary~\ref{cor:theta_moments_asymp}, we obtain the result. 
\qed
\end{proof}

Equation~\eqref{eq:weconcludethat} can be seen as discrete weak formulation for $p(\bx,t,N)$ in space only, and thus
any limiting argument would require some regularity assumption on $p(\bx,t,N)$ in $t$.  In order to circumvent such assumptions,
we need a full discrete weak formulation:

\begin{prop}
 Let $T=M\Delta t$, where $M$ is some fixed positive integer, and let $g$ be an admissible test function, 
with support in $S^{n-1}\times[0,T]$. Let
\[
 \mathsf{T}=\{k\Delta t\},\quad  k=0,\ldots,M-1.
\]
Then we have that
\begin{align}\label{eq:weak_discrete}
&  -\sum_{t\in\mathsf{T}}\sum_{\mathbf{x}\in S^{n-1}_N}p(\mathbf{x},t,N)\left(g(\bx,t+\Delta t)-g(\mathbf{x},t)\right)
-\sum_{\mathbf{x}\in S^{n-1}_N}p(\mathbf{x},t,N)g(\bx,0)\\
\nonumber
&\quad=
\sum_{t\in\mathsf{T}}\sum_{\mathbf{x}\in S^{n-1}_N}p(\mathbf{x},t,N)\left[\frac{1}{2N}\sum_{i,j=1}^{n-1}x_i(\delta_{ij}-x_j)\partial^2_{ij}g(\mathbf{x},t)
 +\left(\Delta t\right)^\nu\sum_{j=1}^{n-1}x_j\partial_{x_j}g(\mathbf{x},t)(\psi^{(j)}(\mathbf{x})-\bar\psi(\mathbf{x}))\right]\\
\nonumber
&\qquad\qquad+\bigO\left(N^{-2}\left(\Delta t\right)^{-1},\left(\Delta t\right)^{2\nu-1},N^{-1}\left(\Delta t\right)^{\nu-1}\right)\nonumber.
\end{align}
\end{prop}
\begin{proof}
Sum \eqref{eq:weconcludethat} over $\mathsf{T}$, and estimate the error term by its total sum, taking into account that
there are $\bigO((\Delta t)^{-1})$ terms in this sum. This shows the right hand side of \eqref{eq:weak_discrete}.
To obtain the left hand side, we perform a summation by parts and use that $g(\bx,T)=0$. See appendix~\ref{ap:A} for details.
\qed
\end{proof}

\subsection{Continuous representation}
\label{ssec:continuos_rep}

The aim is now to obtain a continuous version of \eqref{eq:weak_discrete}, but \textit{without taking any limits yet}.
We first need some preliminary definitions:
\begin{deff}[Piecewise time interpolation]\label{def:pti}
 Let $\mathsf{T}$ be a set of sampling times as above, and let $\mathsf{T}_0$ be a set of times such that
for each $\bar{t}\in\mathsf{T}$, there exists a unique $\xi\in\mathsf{T}_0$ such that  $\xi\in(\bar{t},\bar{t}+\Delta t)$.
Let $g$ be an admissible test function with support in $S^{n-1}\times[0,T]$. Observe that under the assumptions on the
sets $\mathsf{T}$ and $\mathsf{T}_0$, for each $t\in[0,T]$ there exists
a unique $\bar{t}\in\mathsf{T}$ such that $t\in[\bar{t},\bar{t}+\Delta t)$, and a unique $\xi\in (\bar{t},\bar{t}+\Delta t)$.
 With this in mind, we define:
\[
 \hat{g}(\bx,t)=g(\bx,\bar{t}),\quad  t\in[\bar{t},\bar{t}+\Delta t),\quad \bar{t}\in\mathsf{T},
\]
and
\[
\dhat {g}(\bx,t)=g(\bx,\xi),\quad t\in[\bar{t},\bar{t}+\Delta t),\quad  \xi\in(\bar{t},\bar{t}+\Delta t),\quad \bar{t}\in\mathsf{T}\text{ and }\xi\in\mathsf{T}_0.
\]
\end{deff}
\begin{rmk}
For fixed $\bx$, we have on one hand that $\hat{g}(\bx,t)$ is just freezing the value of $g$ on $[\bar{t},\bar{t}+\Delta t)$ to be the
value of $g(\bx,\bar{t})$. On the other hand, $\dhat{g}(\bx,t)$ is freezing the value of $g$ on the same interval to be the value
of $g(\bx,\xi)$, with $\xi\in(\bar{t},\bar{t}+\Delta t)$. The natural choice for $\xi$ will arise, in the present
context, from applications of the mean value theorem to $g$ over the interval $[\bar{t},\bar{t}+\Delta t]$.
\end{rmk}

\begin{deff}[Radonmisation (sic) of discrete densities]\label{def:rdd}
 Let $p(\bx,t,N)$ be a probability density defined no $S^{n-1}_N\times\mathsf{T}$. Let $\delta_{\bx}$ denote
the atomic measure at $\bx$. We define
\[
 p_N(\bx,t)=\sum_{\by\in S^{n-1}_N}p(\by,\bar{t},N)\delta_{\by}(\bx),\quad t\in[\bar{t},\bar{t}+\Delta t).
\]
\end{deff}

With these definitions we have the following  result
\begin{prop}\label{thm:cont_discr}
 Let $g$ be an admissible test function, let $N^{-1}=\kappa\left(\Delta t\right)^\mu$, where $\mu>0$ is a second parameter
yet to be specified, and let 
$p_{\Delta t}(\bx,t)=p_{\kappa^{-1}\left(\Delta t\right)^{-\mu}}(\bx,t)$.
Then there exists a set $\mathsf{T}_0$ as in Definition~\ref{def:pti}, such that
\begin{align}
&-\int_0^\infty\int_{S^{n-1}} p_{\Delta t}(\bx,t)\partial_t\dhat{g}(\bx,t)\,\d\bx\,\d t-\int_{S^{n-1}}p_{\Delta t}(\bx,0)\hat{g}(\bx,0)\,\d\bx\,\d t\nonumber\\
&\quad=\frac{\kappa\left(\Delta t\right)^{\mu-1}}{2}\int_0^{\infty}\int_{S^{n-1}}p_{\Delta t}(\bx,t)\left(\sum_{i,j=1}^{n-1}x_i(\delta_{ij}-x_j)\partial^2_{ij}\hat{g}(\mathbf{x},t)\right)\,\d\bx\,\d t \nonumber \\
&\qquad+\left(\Delta t\right)^{\nu-1} \int_0^{\infty}\int_{S^{n-1}}p_{\Delta t}(\bx,t)\left[\sum_{j=1}^{n-1}x_j\left(\psi^{(j)}(\bx)-\bar\psi(\bx)\right)\partial_{j}\hat{g}(\mathbf{x},t)\right]\d\mathbf{x} \,\d t\label{eqn:cont_discr}\\
&\qquad+\bigO\left(\left(\Delta t\right)^{2\mu-1},\left(\Delta t\right)^{\nu+\mu-1},\left(\Delta t\right)^{2\nu-1}\right).\nonumber
\end{align}

\end{prop}
\begin{proof}
For the right hand side, we observe that $\hat{g}(\bx,t)=g(\bx,t)$ for $\bx\in S_N^{n-1}$ and $k\Delta  t \leq t < (k+1)\Delta t$,
$k=0,1,\ldots$, and that this also holds for all partial derivatives of $g$ not involving $t$. On using the definition of $p_{\Delta t}$, we readily obtain the equivalence between the sums over $S_N^{n-1}$ and the integrals in $\bx$. For the time integrals, we point out that both $p_{\Delta t}(\bx,t)$, $\hat{g}(\bx,t)$ and similarly for the derivatives of $g$ are piecewise constant in $t$. Hence the summation over time can be exactly converted into a time integral with a factor of $(\Delta t)^{-1}$. 
As for the left hand side, apply the mean value theorem to $g(\bx,\cdot)$ to get the result and the set $\mathsf{T}_0$.
 \qed
\end{proof}

\begin{rmk}
 The reader is cautioned that, although \eqref{eqn:cont_discr} has a remarkable resemblance with a weak formulation,
it is not quite so, since the prospective test functions $\hat{g}$ and $\dhat{g}$ are not test functions in the usual sense.
\end{rmk}

\subsection{Passage to the limit}
\label{ssec:passage_to_the_limit}

We now deal with the limit $\Delta t\to0$ in \eqref{eqn:cont_discr}. 

\begin{thm}\label{thm:weak_conv}
 Under the same assumptions of Proposition~\ref{thm:cont_discr}, we have that, for any choice of parameters $\mu$ and $\nu$,
 there exists $p\in L^\infty([0,T],\bm))$, where $\bm$ is the set of positive measures of bounded variation on $S^{n-1}$, such that $p_{\Delta t}(\bx,t)\to p(\bx,t)$ weakly as $\Delta t\to0$. Moreover, the following
limits also hold:
\begin{align*}
& \int_0^\infty\int_{S^{n-1}} p_{\Delta t}(\bx,t)\partial_t\dhat{g}(\bx,t)\,\d\bx\,\d t \to \int_0^\infty\int_{S^{n-1}} p(\bx,t)\partial_tg(\bx,t)\,\d\bx\,\d t\\
&\int_0^{\infty}\int_{S^{n-1}}p_{\Delta t}(\bx,t)\left(\sum_{i,j=1}^{n-1}x_i(\delta_{ij}-x_j)\partial^2_{ij}\hat{g}(\mathbf{x},t)\right)\,\d\bx\,\d t\\
&\qquad\qquad\to  
\int_0^{\infty}\int_{S^{n-1}}p(\bx,t)\left(\sum_{i,j=1}^{n-1}x_i(\delta_{ij}-x_j)\partial^2_{ij}g(\mathbf{x},t)\right)\,\d\bx\,\d t\\
& \int_0^{\infty}\int_{S^{n-1}}p_{\Delta t}(\mathbf{x},t)\left[\sum_{j=1}^{n-1}x_j\left(\psi^{(j)}(\mathbf{x})-\bar\psi(\mathbf{x})\right)\partial_{j}\hat{g}(\mathbf{x},t)\right]\d\bx\,\d t \\
&\qquad\qquad\to
 \int_0^{\infty}\int_{S^{n-1}}p(\mathbf{x},t)\left[\sum_{j=1}^{n-1}x_j\left(\psi^{(j)}(\mathbf{x})-\bar\psi(\mathbf{x})\right)\partial_{j}g(\mathbf{x},t)\right]\d\bx\,\d t
\end{align*}

\end{thm}

\begin{proof}
 From the tightness of Radon measures, cf. \cite{Billingsley:1999},  we have that there exists a sequence $\Delta t_n>0$, with $\Delta t_n \downarrow0$ as $n\to\infty$, and  $p\in L^\infty([0,T],\bm)$, such that
 \[
 \lim_{n\to\infty}p_{\Delta t_n}(\bx,t)=p(\bx,t).
 \] 
The convergence of the integrals follows from the weak convergence of $p_{\Delta t_n}\to p$, and from the fact that for a
continuous function $h$, we have 
\[
 \lim_{\Delta t\to0}\|h-\hat{h}\|_\infty=\lim_{\Delta t\to0}\|h-\dhat{h}\|_\infty=0.
\]
 \qed
\end{proof}

If either $\mu<1$ or $\nu<1$, we can multiply \eqref{eqn:cont_discr} by $(\Delta t)^{-\min(\nu-1,\mu-1)}$. It is then
easily verified that the error term vanishes in the limit, as well as the term with a time derivative. Thus, in this case,
we obtain stationary limits governed by the steady version of the equations derived below. 
Now let us assume that $\mu,\nu\ge 1$.  It is easily verified that the error term will be small. If both $\mu,\nu>1$, we have stationary solutions given by the initial condition. 

The other cases are as follows:
\begin{thm}\label{thm:weak_limits}
There exists $p\in L^\infty([0,T];\bm)$ such that
\begin{description}
 \item If $\mu>1$, $\nu=1$, the \textsl{convective or drift approximation}:
\begin{align}
 &-\int_0^{\infty}\int_{S^{n-1}}p(\bx,t)\partial_tg(\bx,t)\,\d\bx\,\d t - \int_{S^{n-1}}p(\bx,t_0)g(\bx,t_0)\,\d\bx  \nonumber\\
&\quad=\int_0^{\infty}\int_{S^{n-1}}p(\mathbf{x},t)\left[\sum_{j=1}^{n-1}x_j\left(\psi^{(j)}(\mathbf{x})-\bar\psi(\mathbf{x})\right)\partial_{j}g(\mathbf{x},t)\right]\d\mathbf{x}\,\d t.\label{weak:convective}
\end{align}
\item If $\mu=1$, $\nu>1$, the \textsl{diffusive approximation}
\begin{align}
 &-\int_0^{\infty}\int_{S^{n-1}}p(\bx,t)\partial_tg(\bx,t)\,\d\bx\,\d t - \int_{S^{n-1}}p(\bx,t_0)g(\bx,t_0)\,\d\bx \nonumber\\
&\quad=\frac{\kappa}{2}\int_0^{\infty}\int_{S^{n-1}}p(\bx,t)\left(\sum_{i,j=1}^{n-1}x_i(\delta_{ij}-x_j)\partial^2_{ij}g(\mathbf{x},t)\right)\,\d\bx\,\d t.\label{weak:diffusion}
\end{align}
\item If $\mu=1$, $\nu=1$, the case where there is a maximal balance of selection and genetic drift; we find the \textsl{replicator-diffusion equation}
\begin{align}
 &-\int_0^{\infty}\int_{S^{n-1}}p(\bx,t)\partial_tg(\bx,t)\,\d\bx\,\d t - \int_{S^{n-1}}p(\bx,t_0)g(\bx,t_0)\,\d\bx  \nonumber\\
&\quad=\frac{\kappa}{2}\int_0^{\infty}\int_{S^{n-1}}p(\bx,t)\left(\sum_{i,j=1}^{n-1}x_i(\delta_{ij}-x_j)\partial^2_{ij}g(\mathbf{x},t)\right)\,\d\bx\,\d t   \label{weak:replicator_diffusion}\\
&\qquad+ \int_0^{\infty}\int_{S^{n-1}}p(\mathbf{x},t)\left[\sum_{j=1}^{n-1}x_j\left(\psi^{(j)}(\mathbf{x})-\bar\psi(\mathbf{x})\right)\partial_{j}g(\mathbf{x},t)\right]\d\mathbf{x}\,\d t.\nonumber
\end{align}
\end{description}
\end{thm}
\begin{proof} 
The result follows from Theorem~\ref{thm:weak_conv}, and from straightforward bookkeeping of the $\Delta t$ orders
of the terms in \eqref{eqn:cont_discr}. \qed
\end{proof}

\begin{rmk}
Alternatively, Theorems~\ref{thm:weak_conv} and \ref{thm:weak_limits} can be seen together as an existence theorem for equations~\eqref{weak:convective}, \eqref{weak:diffusion} and \eqref{weak:replicator_diffusion}. Under additional regularity hypothesis on the fitness functions we have uniqueness --- see Section~\ref{sec:forward} --- and then the limit is unique. 
\end{rmk}

Equations~(\ref{weak:convective}), (\ref{weak:diffusion}) and~(\ref{weak:replicator_diffusion}) are written in the weak form.
In population dynamics, and in others contexts as well, they are used casted into the strong formulation (or standard PDE formulation) as follows (see, however, remark~\ref{rmk:weak_to_pde}): 

\begin{itemize}
 \item If $\mu>1$ and $\nu=1$, the \textsl{convective of drift approximation}:
\begin{equation}\label{convective_approximation}
 \partial_tp=-\sum_{i=1}^{n-1}\partial_i\left[x_i\left(\psi^{(i)}(\mathbf{x})-\bar\psi(\mathbf{x})\right)p\right] .
\end{equation}
This equation is equivalent to the replicator dynamics, showing that the Wright-Fisher process will be equivalent to the
the replicator dynamics, in the limit of large population and small time-steps, if the population increases faster than
the time-step decreases.
\item If $\mu=1$ and $\nu>1$, the \textsl{diffusive approximation}
\begin{equation}\label{diffusive_approximation}
 \partial_t p=\frac{\kappa}{2}\sum_{i,j=1}^{n-1}\partial_{ij}\left((x_i\delta_{ij}-x_ix_j)p\right) ,
\end{equation}
which is relevant when the fitness converges to 1 as $\Delta t\to 0$ faster than $N\to\infty$.
\item When there is a perfect balance between population size and time step, i.e., $\mu=\nu=1$, we find the \textsl{replicator-diffusion approximation}, given by equation~(\ref{replicator_diffusion_eps}), which we repeat here for convenience:
\begin{equation}\label{replicator_diffusion}\tag{\ref{replicator_diffusion_eps}'}
 \partial_t p=\frac{\kappa}{2}\sum_{i,j=1}^{n-1}\partial_{ij}\left((x_i\delta_{ij}-x_ix_j)p\right)
-\sum_{i=1}^{n-1}\partial_i\left[x_i\left(\psi^{(i)}(\mathbf{x})-\bar\psi(\mathbf{x})\right)p\right] .
\end{equation}
\end{itemize}
We shall focus on the last equation and on its weak formulation~(\ref{weak:replicator_diffusion}).

\begin{rmk}\label{rmk:weak_to_pde}
 We shall see in Section~\ref{sec:forward} that the weak and the PDE formulations are not equivalent, and that the correct formulation is actually the weak one. 
\end{rmk}

\subsection{Conservation laws from the discrete process}
\label{ssec:conservation}

Let us write $\mathfrak{S}$ for the set of all functions $g:S^{n-1}\times[0,+\infty)$ such that there exist an open set
$\bar\Upsilon \supset S^{n-1}$ and a function $G:\bar\Upsilon\times[0,+\infty)\to\mathbb{R}$ such that $g$ is the restriction of $G$ to
$S^{n-1}$ and $G\in C^{2,1}(\bar\Upsilon)$.

Notice that in the right hand side of \eqref{weak:replicator_diffusion}, $p$ is multiplied by:
\begin{equation}\label{adjoint}
\frac{\kappa}{2}\sum_{i,j=1}^{n-1}D_{ij}\partial_{ij}^2g+
\sum_{i=1}^{n-1}\Omega_i\partial_ig=0.
\end{equation}
Equation~\eqref{adjoint} is readily seen to be a steady backward equation. 
We now show that the weak solutions have also conservation laws.
\begin{thm}
 Let $p$ be a solution to \eqref{weak:replicator_diffusion} (we shall take \eqref{weak:diffusion} as a special case). Let $\varphi\in\mathfrak{S}$ be in the kernel of \eqref{adjoint}.
Then
\[
 \int_{S^{n-1}}p(\bx,t)\varphi(\bx)\,\d\bx = \int_{S^{n-1}}p(\bx,0)\varphi(\bx)\,\d\bx,
\]
for almost every $t\in[0,\infty)$.
\end{thm}
\begin{proof}
 Let $\eta(t) \in C_c([0,\infty))$, with $\eta(0)=1$. Then
\[
 g(\bx,t)=\eta(t)\varphi(\bx)
\]
is an admissible test function. On substituting in \eqref{weak:replicator_diffusion}, we find that
\[
 \int_0^\infty\int_{S^{n-1}}p(\bx,t)\varphi(\bx)\eta'(t)\,\d\bx\,\d t + \int_{S^{n-1}}p(\bx,0)\varphi(\bx)\,\d\bx=0.
\]
Since $\eta$ is an arbitrary function with compact support in $[0,\infty)$, the result follows. \qed
\end{proof}

A similar argument shows also the following
\begin{thm}
 Let $p$ be a solution to \eqref{weak:convective}.
Then
\[
 \int_{S^{n-1}}p(\bx,t)\,\d\bx = \int_{S^{n-1}}p(\bx,0)\,\d\bx,
\]
for almost every $t\in[0,\infty)$.
\end{thm}

Therefore, the conservation laws given by equation~(\ref{discrete_conservation_laws}) now become
\begin{equation}\label{eqn:tobe_cons}
 \frac{\d}{\d t}\int_{\slf^{n-1}}p(t,x)\varphi(x)\d x=0,
\end{equation}
where $\varphi$ satisfies \eqref{adjoint}.
In principle the condition set out by \eqref{eqn:tobe_cons} seems to imply an infinite (likely to be uncountable) 
number of conservation laws. The following result shows that it is actually much more conspicuous:
\begin{thm}
\label{thm:finite_cons}
 Let $\mathbf{e}_i$ denote the vertices of $S^{n-1}$. Then there exist unique $\rho_i$, $i=1,\ldots,n-1$,  with $\rho_i(\mathbf{e}_j)=\delta_{ij}$ that are solutions to \eqref{adjoint}. In addition, let $\rho_0\equiv1$. Then, any solution to \eqref{adjoint} in $\mathfrak{S}$ can be written as a linear combination of $\rho_i$, $i=0,\dots,n-1$. In particular its kernel, for solutions in $\mathfrak{S}$, has dimension $n$.
\end{thm}

\begin{proof}
 Given a vertex $\mathbf{e}_i$, let $\mathbf{e}_j$ be an adjacent vertex. Now we solve \eqref{adjoint} in the segment $\overline{\mathbf{e}_j\mathbf{e}_i}$ with boundary values $\delta_{ij}$. In the segments not adjacent to $\mathbf{e}_i$ define the solution to be zero. This defines the solution in all one-dimensional simplices. For each two-dimensional subsimplex, we now solve the Dirichlet problem with the data from the previous step. Now, assume that we have the solution uniquely defined in all subsimplices of dimension $m$. Repeating the construction above yields the solution in all subsimplices of dimension $m+1$. Proceeding inductively, this yields a solution in $S^{n-1}$ that is unique and admissible. Uniqueness follows from the maximum principle applied at each subsimplex level. Let $\varphi\in\mathfrak{S}$, and let
\[
 \Phi(\mathbf{x})=\sum_{i=1}^n\varphi(\mathbf{e}_i)\rho_i(\mathbf{x}).
\]
Then $\Phi\in\mathfrak{S}$. By the preceding argument, $\Phi$ and $\varphi$ must agree at all edges of $S^{n-1}$. Proceeding inductively once again yields that $\varphi=\Phi$ in $S^{n-1}$. \qed
\end{proof}

Thus, the solutions to \eqref{weak:replicator_diffusion} must satisfy:
\begin{equation}\label{continuous_conservation_laws}
\frac{\d}{\d t} \int_{S^{n-1}}\rho_i(\mathbf{x})p(\mathbf{x},t)\d\mathbf{x}=0\ ,\quad i=1,\cdots,n\ .
\end{equation} 

From a probabilistic viewpoint, the $\rho_i$,  $i=1,\ldots,n$, are naturally identified with the fixation probability of type $i$. 
We now give a pure analytical argument for this fact. In Section~\ref{sec:forward}, we shall prove Theorem~\ref{thm:final_state}
which shows that the final state is given by
\begin{equation*}
p^\infty[p^\ini]=\lim_{t\to\infty}p(\cdot,t)=\sum_{i=1}^n\pi_i[p^\ini]\delta_{\mathbf{e}_i}\ ,
\end{equation*}
where $\delta_{\mathbf{e}_i}$ is a Dirac measure supported on the vertex $\mathbf{e}_i\in S^{n-1}_N$. 
Clearly, $\pi_i[p^\ini]$ is the fixation probability of type $i$ in a population initially described by 
a probability distribution $p^\ini$.

Therefore,
\[
 \pi_i[\delta_{\mathbf{x_0}}]=\int\rho_i(\mathbf{x})p^\infty(\mathbf{x})\d\mathbf{x}=\int\rho_i(\mathbf{x})p^\ini(\mathbf{x})\d\mathbf{x}
=\int\rho_i(\mathbf{x})\delta_{\mathbf{x}_0}(\mathbf{x})\d\mathbf{x}=\rho_i(\mathbf{x}_0).
\]

\begin{rmk}
 In the neutral case, i.e., $\psi^{(i)}(\mathbf{x})=\psi^{(j)}(\mathbf{x})$ for all $i,j=1,\dots,n$ and $\mathbf{x}\in S^{n-1}$,
we define the \textsl{neutral fixation probability} $\pi_i^{\mathrm{N}}[\delta_{\mathbf{x}}]=x_i$, which follows from the fact
that in the neutral case, $\rho_i(\mathbf{x})=x_i$.
\end{rmk}

\section{The Replicator-Diffusion approximation}

\label{sec:forward}

We now discuss the nature of solutions $p$ to (\ref{replicator_diffusion}) together with the conservation
laws (\ref{continuous_conservation_laws}). The main result of this section is Theorem~\ref{thm:final_state}. 
This must be understood as the continuous counterpart of
the Lemma~\ref{lem:classical}. We do not refer to the discrete model to prove this result.
Our approach is based solely in the properties of the partial differential equation~(\ref{replicator_diffusion}), the restriction
of the domain to the domain of interest, and
the associated conservation laws~(\ref{continuous_conservation_laws}). 

An outline of the proof of Theorem~\ref{thm:final_state} is as follows:
First, we show that a solution to \eqref{weak:replicator_diffusion}  can be written as regular part plus a singular measure over the boundary. Moreover, the regular part vanishes for large time. Repeating these arguments over the lower dimensional subsimplices, and using the projection result in Proposition~\ref{prop:face_proj},  we arrive at a representation of $p$ as a sum of its classical solution and
a sum of singular measures that are uniformly supported on the descending chain of subsimplices of $S^{n-1}$ down to the zeroth dimension. 
Since the solutions over the subsimplices also have a regular part that vanishes, we can show that all measures that are not atomically supported at the vertices should vanish for large time. Thus, conservation of probability implies that the steady state of \eqref{weak:replicator_diffusion} is a sum of deltas.

Finally, we provide two applications.
In Subsection~\ref{ssec:duality}, we study the dual equation. This will be the continuous limit of the 
 evolution by the dual equation (backward equation) of the discrete process and therefore its solution $f(\mathbf{k},t)$ gives the fixation probability at time $t$ of a given type (to be prescribed by the boundary conditions in the dual process) for a population initially at state $\mathbf{k}$. This gives a generalization for an arbitrary number of types and for arbitrary fitnesses
of the celebrated Kimura equation with reversed time~\citep{Kimura}.  In the sequel, Subsection~\ref{ssec:strat_domin},
we will show that if one type dominates all other types then, for any initial condition, the fixation probability of this
type will be larger than the neutral fixation probability. This shows, in particular, that for large populations, the
most probable type to fixate will be the one playing the Nash-equilibrium strategy of the game (assuming the identity
between fitness and pay-offs, which is standard in this framework). This is not true in general for small populations~\citep{Nowak_2006}.

\subsection{Solution of the replicator-diffusion equation}
\label{ssec:solution}

We now study in more detail the features of the solution to (\ref{weak:replicator_diffusion}) and show two important results: first that in the interior of the simplex, the solution must satisfy (\ref{replicator_diffusion}) in the classical sense; second, no classical solution to (\ref{replicator_diffusion}) can satisfy the conservation laws.
Throughout this section, we shall have the further assumption that the fitnesses are smooth.

We start by showing that, in the interior, a weak solution is regular enough to be a classical solution.

\begin{lem}
\label{lem:classical_int}
 Let $p$ be a solution to (\ref{weak:replicator_diffusion}). Let $K\subset S^{n-1}$ be a proper compact subset. Then, in $K$, $p$ satisfies (\ref{replicator_diffusion}) in the classical sense. In particular, $p\in C([0,T];C^{\infty}(\mathop{\mathrm{int}}(S)))$.
\end{lem}

\begin{proof}
Let  $g \in C_c^\infty(K)$, we have then the standard weak formulation of (\ref{replicator_diffusion}) in $K$.  On the other hand, (\ref{replicator_diffusion}) is uniformly parabolic in any proper subset. Hence the weak and strong formulations coincide --- c.f. \citep{Evans,Taylor96a}.
The last statement follows from $\mathop{\mathrm{int}}(S)=\cup_{K\subset S}K$, with $K$ compact and $K\cap\partial S^{n-1}=\emptyset$. \qed
\end{proof}

The next two Lemmas show existence of a unique classical solution, and that such a solution decays to zero for large time.

\begin{lem}
\label{lem:ncons_pde}
 Let $p$ be a classical solution to (\ref{replicator_diffusion}). Then
\[
 \lim_{t\to\infty} p(\bx,t)=0,\quad \bx\in \mathop{\mathrm{int}} S^{n-1}.
\]
\end{lem}
\begin{proof}
 We define $\mu_{\mathrm{S}}(\mathbf{x})=x_1x_2\cdots x_n$
(such that $\mu_{\mathrm{S}}(\mathbf{x})\ge 0$ in $S^{n-1}$ with $\mu_{\mathrm{S}}=0$ if and only if $\bx\in\partial S^{n-1}$). 
Note that
\begin{equation}
 \sum_{j=1}^{n-1}\partial_j\left(\frac{D_{ij}}{\mu_{\mathrm{S}}}\right)=\mu_{\mathrm{S}}^{-1}\left[\sum_{j=1}^{n-1}\left(\delta_{ij}-\delta_{ij}x_j-x_i\right)
-\sum_{j=1}^{n-1}\left(x_j\delta_{ij}-x_ix_j\right)\left(\frac{1}{x_j}-\frac{1}{x_n}\right)\right]=0\ .
\label{change_zero}
\end{equation}
We introduce the new variable $u=\mu_{\mathrm{S}}p$ and
after some manipulations, we find
\begin{equation}
\partial_t u=\mu_{\mathrm{S}}\nabla\cdot\left[\mu_{\mathrm{S}}^{-1}\left(\frac{\kappa}{2}D\nabla u-\boldsymbol{\Omega} u\right)\right]\ ,
\label{RD_chang_var}
\end{equation}
with $D=\left(D_{ij}\right)_{i,j=1,\dots,n-1}$ and $\boldsymbol{\Omega}=\sum_{i=1}^{n-1}\Omega_i\mathbf{e}_i$.

We now show that the last equation is well defined in $\slf^{n-1}$. For the second order term, this follows from a new application
of equation~(\ref{change_zero}). For the first order term, note that
\[
 \mu_{\mathrm{S}}\nabla\cdot\left(\frac{\boldsymbol{\Omega}u}{\mu_{\mathrm{S}}}\right)=
\nabla\cdot\boldsymbol{\Omega} u-\frac{\boldsymbol{\Omega}\cdot\nabla\mu_{\mathrm{S}}}{\mu_{\mathrm{S}}}u+\boldsymbol{\Omega}\cdot\nabla u\ .
\]
Furthermore
\[
\frac{\boldsymbol{\Omega}\cdot\nabla\mu_{\mathrm{S}}}{\mu_{\mathrm{S}}}=\sum_{i=1}^{n-1}\Omega_i\left(\frac{1}{x_i}-\frac{1}{x_n}\right)=
\sum_{i=1}^n\left(\psi^{(i)}(\bx)-\bar\psi(\bx)\right)\ .
\]

We shall now  study the eigenvalue problem associated to \eqref{RD_chang_var} by
 considering the dual problem, with respect to the measure $(\mu_S)^{-1}\d\bx $, and with regularised coefficients: 
\begin{equation}
\left\{\begin{array}{ll}
&\mu_{\mathrm{S}}^{(\eps)}\nabla\cdot\left[\frac{\kappa}{2\mu_{\mathrm{S}}^{(\eps)}}D^{(\eps)}\nabla \varphi^{(\eps)}\right] + s \boldmath{\Omega}\cdot\nabla\varphi^{(\eps)}=
\lambda^{(\eps)}\varphi^{(\eps)},\\
&\varphi^{(\eps)}=0\text{ in }\partial\slf^{n-1}\ ,
\end{array}\right.
\label{for:evp2}
\end{equation}
where $D^{(\eps)}(\bx)$ is a positive defined matrix in $S^{n-1}$, with $D^{(\eps)}\stackrel{\eps\to0^+}{\longrightarrow}D$ uniformly in $\bx$, and  $\mu^{(\eps)}_{\mathrm{S}}>0$ in $\slf^{n-1}$, $\mu^{(\eps)}_{\mathrm{S}}\stackrel{\eps\to0^+}{\longrightarrow}\mu_{\mathrm{S}}$ uniformly in $\bx$, and $s$ is a real parameter. 

First we observe that, for $\eps\geq0$, equation \eqref{for:evp2} satisfies a maximum principle for solutions in $C^2(\mathop{\mathrm{int}}(S^{n-1}))$; therefore if $\lambda^{(\eps)}=0$ we have $\varphi^{(\eps)}=0$ in $\slf^{n-1}$; see \cite{Crandalletal1992}. We conclude that $\lambda^{(\eps)}\not=0$, for  $s\in\R$ and $\eps\geq0$. Additionally, since the coefficients are smooth, the solution to equation \eqref{for:evp2} is smooth in the interior by standard elliptic regularity.

For $\eps>0$, the dominant eigenvalue $\lambda_0^{(\eps)}$ is real and from the maximum principle it follows that
$\lambda_0^{(\eps)}\not=0$. For $s=0$, $\lambda_0^{(\eps)}$ is negative and therefore from its continuity in $s$, we conclude that $\lambda^{(\eps)}_0<0$ for any $s$. Therefore,
for any other eigenvalue $\mathrm{Re}\left(\lambda^{(\eps)}\right)\le\lambda^{(\eps)}_0<0$ (see~\citep{Evans} for further details).

Moreover, let $\eps_k\to0$ be a decreasing sequence of positive numbers, and $\varphi^{(\eps_k)}\ge 0$ be the normalised eigenfunctions for the corresponding leading eigenvalues. Since the coefficients are assumed smooth, the eigenfunctions are also smooth. Hence, by Rellich theorem, there is a subsequence $\eps_{k_j}$ such that $\varphi^{(\eps_{k_j})}$ converges in $L^2(S^{n-1})$. By considering the weak formulation for equation~\eqref{for:evp2}, we immediately see that, for this subsequence, we must also have $\lambda^{(\eps_{k_j})}_0\to\lambda_0$.

We thus have obtained a real negative eigenvalue $\lambda_0$, with a real eigenfunction that is single signed. Using the same argument  in \cite{Evans}, we conclude that $\lambda_0$ is the principal eigenvalue for the nonregularised problem. Hence, any other eigenvalue will satisfy $\mathop{\mathrm{Re}}(\lambda)\leq\lambda_0$.

This also shows that there exists $\alpha>0$, such that 
\[
\frac{1}{2}\partial_t\int_{S^{n-1}} u^2\mu_S^{-1}d\bx=
\int_{S^{n-1}}\mu_{\mathrm{S}}\nabla\cdot\left[\mu_{\mathrm{S}}^{-1}\left(\frac{1}{2}D\nabla u - \boldsymbol{\Omega}u\right)\right] u \,\mu_{\mathrm{S}}^{-1}\d\bx
<-\alpha \int_{S^{n-1}}u^2\,\mu_{\mathrm{S}}^{-1}\d\bx.
\]

Therefore
\[
\int p^2\mu_{\mathrm{S}}\d\bx=\int u^2\mu_{\mathrm{S}}^{-1}\d\bx\stackrel{t\to\infty}{\rightarrow}0\ .
\]
\qed
\end{proof}

\begin{lem}\label{lem:unique}
 Equation~(\ref{RD_chang_var}) has a unique solution $u\in C\left([0,\infty);C^{\infty}(\mathrm{int}\left(S^{n-1}\right)\right)$.
\end{lem}

\begin{proof}
Consider equation~(\ref{RD_chang_var}) with $D=D^{(\eps)}$ and $\mu_{\mathrm{S}}=\mu_{\mathrm{S}}^{(\eps)}$, with $D^{(\eps)}$ and $\mu_{\mathrm{S}}^{(\eps)}$ as in Lemma~\ref{lem:ncons_pde}.
For $\eps>0$, it is uniformly parabolic, and hence it has a unique solution with the required regularity. 

We write (\ref{RD_chang_var}) in weak form as
\begin{align*}
& \int_0^\infty\int_{\slf^{n-1}}u^{(\eps)}(t,x)\partial_t\phi(t,x)\left(\mu_{\mathrm{S}}^{(\eps)}\right)^{-1}\d\bx\d t\\
&\qquad
+\int_0^\infty\int_{\slf^{n-1}}\left(\mu_{\mathrm{S}}^{(\eps)}\right)^{-1}\left(\frac{1}{2}D\nabla u^{(\eps)}-\boldmath{\Omega}u^{(\eps)}\right)\cdot\nabla\phi(t,x)\d\bx\d t\\
&\qquad+\int_{\slf^{n-1}}u^{(\eps)}(0,x)\phi(0,x)\left(\mu_{\mathrm{S}}^{(\eps)}\right)^{-1}\d\bx=0.
\end{align*}
We now observe that any such solution is bounded in $\aleph=L^2((0,T);W^{2,1}_0)$. Hence,  one can select a sequence $\eps_k\downarrow0$ such that $u^{(\eps_k)}\to u^{*}\in\aleph$.  Since \eqref{RD_chang_var} is weakly parabolic,  such a solution must be unique --- \cite{Lieberman1996}. Finally, regularity follows from Lemma~\ref{lem:classical_int}.
\qed
\end{proof}

The preceding two lemmas have an important consequence:

\begin{cor}
\label{cor:no_cons}
 No solution to (\ref{replicator_diffusion}) in the classical sense can satisfy the required conservation laws. In particular, this shows that the weak-formulation presented in Theorem~\ref{thm:weak_limits} is not only a device for obtaining the continuous limit, but it turns out to be the correct formulation.
\end{cor}

In view of Corollary~\ref{cor:no_cons}, we turn back to the weak formulation given by \eqref{weak:replicator_diffusion} or, equivalently, to \eqref{replicator_diffusion} together with the conservation laws \eqref{continuous_conservation_laws}. In what follows, we obtain some more information about such solutions.

\begin{rmk}\label{rmk:decomp}
 As an extension of Lemma~4.1 in \cite{Chalub_Souza:CMS_2009}, we observe that if $p$ is a Radon measure in $S^{n-1}$,
we can write $p=q+r$, with $\sing(q)\in\partial S^{n-1}$ and $\sing(r)\in \mathop{\mathrm{int}} S^{n-1}$. 
\end{rmk}

\begin{prop}[Face Projections]
\label{prop:face_proj}
 Let $1\leq k<n$ and let $p$ be a solution to \eqref{weak:replicator_diffusion}.  Let $S^{k-1}$ be a face of $S^{k}$. Assume that $\sing(p)\cap S^{k-1}\not=\emptyset$. Then, over $S^{k-1}$, $p$ satisfies \eqref{weak:replicator_diffusion} in one less dimension
with forcing given by the regular part of $p$ evaluated at $x_i=0$, for a certain value of $i$.
\end{prop}

\begin{proof}
Assume, without loss of generality, that $i=1$.
 In view of remark~\ref{rmk:decomp}, we can write $p=q+r$, where $\sing(q)\subset \slf^{n-2}$  with the singular support of $r$ lying in the complement with respect to the full simplex. Moreover, we can also assume, without loss of generality, that $S^{n-2}$ is given by the intersection of the hyperplane $x_1=0$ with $S^{n-1}$. Let us write $\bx=(x_2,\ldots,x_n)$. Let $h$ be an appropriate test function in $S^{n-2}$, satisfying $h(\bx,0)=0$ and let $\eta(x_1)\in C_c([0,1])$, with $\eta(0)=1$. Then  $g=\eta h$ is an appropriate test function for $S^{n-1}$ and a direct computation with \eqref{weak:replicator_diffusion} then yields
\begin{align*}
&-\int_0^{\infty}\int_{S^{n-1}}p(x_1,\bx,t)\partial_tg(x_1,\bx,t)\,\d\bx\,\d t\\ 
&\quad=\frac{\kappa}{2}\int_0^{\infty}\int_{S^{n-1}}p(x_1,\bx,t)\left(\sum_{i,j=1}^{n-1}x_i(\delta_{ij}-x_j)\partial^2_{ij}g(x_1,\mathbf{x},t)\right)\,\d\bx\,\d t  \\
&\qquad+ \int_0^{\infty}\int_{S^{n-1}}p(x_1,\mathbf{x},t)\left[\sum_{j=1}^{n-1}x_j\left(\psi^{(j)}(x_1,\mathbf{x})-\bar\psi(x_1,\mathbf{x})\right)\partial_{j}g(x_1,\mathbf{x},t)\right]\d\mathbf{x}.
\end{align*}
Over $x_1=0$, on using the definition of $g$, we find that
\begin{align*}
 -\int_0^{\infty}\int_{S^{n-2}}q(\bx,t)\partial_th(\bx,t)\,\d\bx\,\d t 
&=\frac{\kappa}{2}\int_0^{\infty}\int_{S^{n-2}}q(\bx,t)\left(\sum_{i,j=2}^{n-1}x_i(\delta_{ij}-x_j)\partial^2_{ij}h(\mathbf{x},t)\right)\,\d\bx\,\d t \\
&\quad+\int_0^{\infty}\int_{S^{n-2}}q(\mathbf{x},t)\left[\sum_{j=2}^{n-1}x_j\left(\psi^{(j)}(\mathbf{x})-\bar\psi(\mathbf{x})\right)\partial_{j}h(\mathbf{x},t)\right]\d\mathbf{x}\,\d t.
\end{align*}
For $r$, we have
\begin{align*}
&-\int_0^{\infty}\int_{S^{n-1}}r(x_1,\bx,t)\partial_tg(x_1,\bx,t)\,\d\bx\,\d t\\ 
&\quad=\frac{\kappa}{2}\int_0^{\infty}\int_{S^{n-1}}r(x_1,\bx,t)\left(\sum_{i,j=1}^{n-1}x_i(\delta_{ij}-x_j)\partial^2_{ij}g(x_1,\mathbf{x},t)\right)\,\d\bx\,\d t  \\
&\qquad+ \int_0^{\infty}\int_{S^{n-1}}r(x_1,\mathbf{x},t)\left[\sum_{j=1}^{n-1}x_j\left(\psi^{(j)}(x_1,\mathbf{x})-\bar\psi(x_1,\mathbf{x})\right)\partial_{j}g(x_1,\mathbf{x},t)\right]\d\mathbf{x}\,\d t.
\end{align*}

By Lemma~\ref{lem:classical_int}, $r$ is smooth. Therefore, the above equation can be integrated by parts to yield,
an integral on $S^{n-1}$ that will cancel out identically, since $r$ is a classical solution to \eqref{replicator_diffusion}, and
a number of integrals over the various faces of $S^{n-1}$. In particular, at $x_1=0$, we find that
\[
 0=-\frac{\kappa}{2}\int_0^{\infty}\int_{S^{n-2}}r(0,\bx,t)h(\bx,t)\,\d\bx\,\d t.
\]
By collecting together the two calculations on $x_1=0$, we obtain  the result.
\qed
\end{proof}

In what follows, we  shall need some preliminaries. Recall --- see \cite{Stanley1996} --- that to the simplex $\slf^{n-1}$ is associated a corresponding  $f$-vector, such that the entry $i+1$ ($f_{i+1}$)  is the number of $i$-dimensional subsimplices of $\slf^{n-1}$.  We shall assume that, for each dimension $i$, there is a definite order of the subsimplices $\slf^{i,j}$, with $i=0,\ldots,n-1$ and $j=1,\ldots,f_{i+1}$. Moreover, we define the adjacent operator by $\ad(j,k)$ which denotes the $k$th adjacent subsimplex of dimension $i+1$ to $\slf^{ij}$. Notice that there are $n-i$ such simplexes.

\begin{thm}[Solution Structure]
\label{thm:soln_struct} 
Equation \eqref{weak:replicator_diffusion}, with a given initial condition  $p^\ini\in\bm$, has a unique solution $p\in L^\infty\left([0,T];\bm\right)$.  Moreover,  let $\delta^{ij}$ be the Radon measure with unit mass uniformly supported on $\slf^{ij}$. Then the solution $p$ can be written as
\begin{equation}
p(t,x)=p_{n1}+\sum_{\left(i,j\right)\in\Ind}p_{ij}\delta^{ij},\quad \Ind=\left(\cup_{i=0}^{n-1}\{i\}\right)\times\{1,\ldots,f_{i+1}\},
 \label{strut_soln}
\end{equation}
where $p_{ij}$ satisfies 
\begin{align}
 &-\int_0^{\infty}\int_{S^{ij}}p_{ij}(\bx,t)\partial_tg(\bx,t)\,\d\bx\,\d t - \int_{S^{ij}}p_{ij}(\bx,t_0)g(\bx,t_0)\,\d\bx  \nonumber\\
&\quad=\frac{\kappa}{2}\int_0^{\infty}\int_{S^{ij}}p_{ij}(\bx,t)\left(\sum_{r,s=1}^{n-1}x_r(\delta_{rs}-x_s)\partial^2_{rs}g(\mathbf{x},t)\right)\,\d\bx\,\d t   \label{weak:rd_forcing}\\
&\qquad+ \int_0^{\infty}\int_{S^{ij}}p_{ij}(\mathbf{x},t)\left[\sum_{s=1}^{n-1}x_s\left(\psi^{(s)}(\mathbf{x})-\bar\psi(\mathbf{x})\right)\partial_{s}g(\mathbf{x},t)\right]\d\mathbf{x}\,\d t\nonumber\\
&\qquad \quad + \int_0^{\infty}\int_{S^{ij}}\sum_{k=1}^{n-i}\left.p_{(i+1)\ad(j,k)}\right|_{S^{ij}}\,\d\bx\,\d t.\nonumber
\end{align}
The initial condition for \eqref{weak:rd_forcing} will be denoted $p^\ini_{ij}$, and it  is obtained from $p^\ini$ by applying the decomposition described in Remark~\ref{rmk:decomp} recursively.
\end{thm}

\begin{proof}
By direct substitution into  \eqref{weak:replicator_diffusion} and after a integrating by parts starting from $\slf^{n-1}$, proceeding downwards until the vertices and using Lemma~\ref{prop:face_proj} one can verify that \eqref{strut_soln} is indeed a solution. To verify uniqueness, let $\tp$ a solution to \eqref{weak:replicator_diffusion}.
Consider  $\tp$ restricted to  $S^{n-1}$ and let $p_{n1}$ be the classical solution guaranteed by Lemma~\ref{lem:unique}.
By considering \eqref{weak:replicator_diffusion} with test functions with compact support on $\mathop{\mathrm{int}} S^{n-1}$, we see that $\tp-p_{n1}$ vanishes. Therefore $\sing(\tp-p_{n1})\subset \partial S^{n-1}$. By Remark~\ref{rmk:decomp}, we can write $\tp=p_{n1}+q$, with $\sing(q)\subset\partial S^{n-1}$.  Now $\partial S^{n-1}$ is the union of
$f_{n-1}$ copies of $S^{n-2}$. By Proposition~\ref{prop:face_proj}, $q$ must satisfies \eqref{weak:rd_forcing} in one
less dimension in each of the subsimplices. Proceeding inductively, we can now choose a subsimplex  $S^{n-3}$ of $S^{n-2}$. Now, we repeat the argument above for each simplex $S^{n-2}$ which has $S^{n-3}$ as a subsimplex.
Iterate until arrive at the simplices of zero dimension to get the result.\qed
\end{proof}

This theorem leads to the following result:

\begin{thm}[Final State]
\label{thm:final_state}
Let
\[
 p^\infty(\bx)\bydef\lim_{t\to\infty}p(\bx,t)\ ,
\]
where $p$ is the solution of equation~(\ref{replicator_diffusion_eps}) subject to conservation laws~(\ref{continuous_conservation_laws}).
Then $p^\infty$ is a linear combination of point masses at the vertices of $S^{n-1}$,i.e,
\begin{equation}\label{final_state}
 p^\infty=\sum_{i=1}^n \pi_i\left[p^\ini\right]\delta_{\mathbf{e}_i}\ .
 \end{equation}
\end{thm}

\begin{proof}
First, we observe that Lemma~\ref{lem:ncons_pde} still holds if applied to inhomogeneous version of \eqref{RD_chang_var}, provided
that the forcing decays for large times. The results now follows from a straightforward application of Proposition~\ref{prop:face_proj} together with Lemma~\ref{lem:ncons_pde} applied in a descending chain of simplices down to dimension 1.
Conservation of probability then yields that $p^\infty$ must be a sum of atomic measures at the vertices of $S^{n-1}$. On using the other conservation laws, we obtain the coefficients, and hence the result.
\qed
\end{proof}

\subsection{Duality and the Kimura equation}
\label{ssec:duality}

The formal adjoint of equation~(\ref{replicator_diffusion_eps}) (changing the flow of time from forward to backward) provides a generalization of the celebrated
Kimura equation~\citep{Kimura}, both including more types and allowing frequency dependent fitness: 
\begin{equation}\label{backward_kimura_full}
\partial_tf=\mathcal{L}^\dagger_{n-1,k}f\bydef\frac{\kappa}{2}\sum_{i,j=1}^{n-1}D_{ij}\partial_{ij}^2f+\sum_{i=1}^{n-1}\Omega_i\partial_if\ .
\end{equation}
In diffusion theory this equation is associated with a martingale problem for the diffusive continuous process. 
In genetics, the meaning of equation~(\ref{backward_kimura_full})  is seldom made clear and depends on the boundary conditions imposed. One possible and common interpretation is as follows: given an homogeneous state $\mathbf{e}_i\in\Delta S^{n-1}$, let 
$f_i(\mathbf{k},t)$ be the probability that given a population initially in a well-defined state
$\mathbf{k}\in S^{n-1}$ (i.e., $p^\ini(\mathbf{x})\bydef p(\mathbf{x},0)=\delta_{\mathbf{k}}(\mathbf{x})$)
we find the population fixed at the homogeneous state $\mathbf{e}_i$ at time $t$ (or before), 
i.e., 
$f_i(\mathbf{k},t)=\langle p(\cdot,t),\delta_{\mathbf{e}_i}\rangle$.
In this case, we need to find consistent boundary conditions. See~\citet{Maruyama,EtheridgeLNM}.

Let us study the fixation of type 1, represented by the state $\mathbf{e}_1$.
Let us now call $V_i$ the face of the
simplex with $x_i=0$ (type $i$ is absent). Then, $f_i\big|_{V_1}=0$. For $i\ne 1$, 
$f_i\big|_{V_i}$ is the solution of $\partial_t f=\mathcal{L}^\dagger_{n-2,k} f$,
where the type $i$ was omitted from the equation. As the faces of the simplex are
invariant under the adjoint evolution (one more fact to be attributed to lack of mutations
in the model), this represent the same problem in one dimension less. We continue
this procedure until we find the evolution in the edge
from vertex $1$ to vertex $i\ne 1$, $L_{1i}$. In this case, we have that $f\big|_{L_{1i}}:[0,1]\to\mathbb{R}$, the
restriction of $f_i$ to this edge, with $x_1$ being the fraction of type 1 individuals, is the solution
of 
\begin{equation}\label{Kimura}
\partial_t f=\frac{\kappa}{2}x_1(1-x_1)\partial^2_1f+
x_1(1-x_1)\left(\psi^{(1)}_{1i}(x_1)-\psi^{(i)}_{1i}(x_1)\right)\partial_1 f
\end{equation}
with boundary conditions given by $f(0)=0$ and $f(1)=1$ and $\psi^{(j)}_{1i}(x_1)=\psi^{(j)}(x_1\mathbf{e}_1+(1-x_1)\mathbf{e}_i)$ is the restriction of
$\psi^{(j)}$ to the edge $L_{1i}$.
The forward and backward versions of 
Equation~(\ref{Kimura}) are fully studied in the references~\citep{Chalub_Souza:CMS_2009,Chalub_Souza:TPB_2009}. 
For $\psi^{(1)}_{1i}-\psi^{(i)}_{1i}$ constant this is the 
Kimura equation.

\subsection{Strategy dominance}
\label{ssec:strat_domin}

Let us assume that $\psi^{(1)}(\mathbf{x})\ge\psi^{(i)}(\mathbf{x})$ for all $\mathbf{x}\in S^{n-1}$. This happens, for example, if we identify fitness functions with pay-offs in game theory, types with strategists, and if strategist 1 plays the Nash-equilibrium strategy.

Therefore, we prove
\begin{thm}
 If, for all states $\mathbf{x}\in S^{n-1}$, and all types $i=1,\dots,n$, $\psi^{(1)}(\mathbf{x})\ge\psi^{(i)}(\mathbf{x})$, then the
fixation probability of the first type is not less than the neutral fixation probability for any initial condition $p^\ini$; i.e,
\[
 \pi_1[p^\ini]\ge\pi_1^{\mathrm{N}}[p^\ini]\ .
\]
\end{thm}
 
\begin{proof}
 First note that it is enough to prove that $\pi_1[\delta_{\mathbf{x}}]\ge\pi_1^{\mathrm{N}}[\delta_{\mathbf{x}}]=x_1$ for all 
$\mathbf{x}\in S^{n-1}$. The difference $\rho_1(\mathbf{x})-x_1$ satisfy
\[
 \frac{\kappa}{2}\sum_{i,j=1}^{n-1}D_{ij}\partial^2_{ij}\left(\rho_1(\mathbf{x})-x_1\right)+\sum_{i=1}^{n-1}\Omega_i\partial_i\left(\rho_1(\mathbf{x})-x_1\right)
=-\Omega_1=-x_1\left(\psi^{(1)}(\mathbf{x})-\bar\psi(\mathbf{x})\right)\le 0\ ,
\]
with vertex conditions $\rho_1(\mathbf{e}_i)-x_1(\mathbf{e}_i)=0$ for $i=1,\dots,n$.
Now, we proceed by induction in $n$. For the case $n=2$, the proof is in~\cite[Section 4.3]{Chalub_Souza:TPB_2009}; we reproduce it
here only for completeness. 

We write explicitly the equation for $\rho_1$:
\[
 \frac{\kappa}{2}x(1-x)\partial_x^2\rho_1+x\left(\psi^{(i)}(x)-\bar\psi(x)\right)\partial_x\rho_1=0
\]
with $\rho_1(0)=0$ and $\rho_1(1)=1$. We simplify the equation using the fact that 
$\psi^{(1)}(x)-\bar\psi(x)=(1-x)\left(\psi^{(1)}(x)-\psi^{(2)}(x)\right)$ and the solution is given by
\[
 \rho_1(x)=\frac{\int_0^x\exp\left[-\frac{2}{\kappa}\int_0^{\bar x}\left(\psi^{(1)}(\bar{\bar x})-\psi^{(2)}(\bar{\bar x})\right)\d \bar{\bar x}\right]\d \bar x}
{\int_0^1\exp\left[-\frac{2}{\kappa}\int_0^{\bar x}\left(\psi^{(1)}(\bar{\bar x})-\psi^{(2)}(\bar{\bar x})\right)\d \bar{\bar x}\right]\d \bar x}\ .
\]
As $\psi^{(1)}(x)\ge\psi^{(2)}(x)$, we conclude that 
\begin{align*}
&\frac{1}{x}\int_0^x\exp\left[-\frac{2}{\kappa}\int_0^{\bar x}\left(\psi^{(1)}(\bar{\bar x})-\psi^{(2)}(\bar{\bar x})\right)\d \bar{\bar x}\right]\d \bar x
\\&\qquad\qquad\ge
\int_0^1\exp\left[-\frac{2}{\kappa}\int_0^{\bar x}\left(\psi^{(1)}(\bar{\bar x})-\psi^{(2)}(\bar{\bar x})\right)\d \bar{\bar x}\right]\d \bar x\ .
\end{align*}
In particular, $\rho_1(x)\ge x$.

Now, assume that
$\rho_1(\mathbf{x})-x_1\ge 0$ for all $\mathbf{x}\in\partial S^{n-1}$. (Note that $\partial S^{n-1}$ is an union of a finite
number of $n-2$ dimensional simplexes, where by the principle of induction we assume the result valid.) Finally, we use the maximum
principle for subharmonic functions to conclude that the minimum cannot be in the interior of
the simplex~\citep{CourantHilbert2}. Therefore $\rho_1(\mathbf{x})\ge x_1$ for all $\mathbf{x}\in S^{n-1}$. 
\qed
\end{proof}

\section{The Replicator Dynamics}
\label{sec:replicator}

In Section~\ref{sec:infinite}, we proved that, when genetic drift and selection balance, then there is a special timescale such that the evolution of an infinite population can be described by a parabolic partial differential equation. Nevertheless, in applications one is usually interested in large but finite populations. In this case, an exact limit is not taken, and \eqref{weak:replicator_diffusion} can be taken as an approximation of this evolution. We shall discuss this further in the conclusions, but we observe that this equation might be a good approximation even when balance is not exact, i.e., when $\nu$ and $\mu$ are close but not equal to one. This could typically lead to an equation with $\kappa$ being either quite large or small. In the former case, a regular expansion in $\kappa$ shows that the evolutions is governed by \eqref{weak:diffusion}. On the other hand, in the latter case, one expects that the much simpler transport equation \eqref{convective_approximation} will be a good approximation 
for the evolution. Indeed,
in this section we show that \eqref{replicator_diffusion} can be uniformly approximated  by \eqref{convective_approximation}
in proper compact subsets of the simplex, and over a time interval shorter than $\kappa^{-1}$.

We start in Subsection~\ref{ssec:repODE} showing that the equation~(\ref{convective_approximation})
is formally equivalent to the replicator system. Afterwards,
in Subsection~\ref{ssec:peak}, we answer what we believe to be an important question: what exactly is the replicator equation modelling? In particular,
we will show, using a simple argument, that the replicator equation does not model the evolution of
the expected value (of a given trait) in the population, but the evolution of the most common trait  
conditional on the absence of extinctions. 
Finally, we show, in Subsection~\ref{ssec:local}, that the replicator ordinary differential
equation is a good approximation for the initial dynamics of the Wright-Fisher process, when 
$\kappa$ is small. As in Section~\ref{sec:forward}, we shall assume that the fitness functions are smooth.

\subsection{The replicator ODE and PDE}
\label{ssec:repODE}

We shall now study in more detail the equation~(\ref{convective_approximation}), which has 
a close connection with the replicator dynamics as shown below:
\begin{thm}\label{thm:replicator_eu}
Assume that $\Omega$ is Lipschitz.
 Let $\Phi_t(\bx)$ the flow map of 
\begin{equation}\label{replicator_ode}
 \frac{\d\mathbf{x}}{\d t}=\boldsymbol{\Omega}(\bx(t)).
\end{equation}
and let 
\[
 Q(\bx,t)=-\int_0^t(\nabla\cdot\boldsymbol{\Omega})(\Phi_{s-t}(\bx))\d s.
\]
Let $p^I\in\bm$ and assume that $\sing(p^I)\subset \mathop{\mathrm{int}}(S^{n-1})$ (see Remark~\ref{rmk:decomp}).
Then the solution to \eqref{convective_approximation} with initial condition $p^I$ is given by
\begin{equation}
p(\bx,t)=\e^{Q(\bx,t)}p^\ini\left(\Phi_{-t}(\bx)\right).
 \label{replicator_pde:soln}
\end{equation}
\end{thm}
\begin{proof}
We observe that, since $\Omega$ is Lipschitz, the push-forward of $p^I$ by $\Phi_t$ is well defined---\citep{ambrosioetal2005}. Hence, 
the proof is based on the methods of characteristics; see~\cite{John_F_PDE,Evans}. See \cite{DipernaLions1989} for an approach that works even if $\Omega$ fails to be Lipschitz continuous.
\qed
\end{proof}

\begin{rmk}
If $p^\ini$ gives mass to the boundaries of $S^{N-1}$, we write, as in Theorem~\ref{thm:soln_struct}:
\[
p^\ini=\sum_{i,j}p^\ini_{i,j},\quad\text{with}\quad \sing(p^\ini_{i,j})\subset S^{i,j}.
\]
Moreover, notice that $\Omega$ restricted to $S^{i,j}$ is tangent to $\partial S^{i,j}$. Hence, the restricted dynamics is always well defined, and we can write $\Phi^{i,j}_t$ for the flow map of 
\[
\frac{\d\bx^{i,j}}{\d t}=\Omega^{i,j}(\bx^{i,j}(t))
\quad\text{and}\quad
Q^{i,j}(\bx,t)=-\int_0^t\left(\nabla^{i,j}\cdot\Omega^{i,j}\right)\left(\Phi^{i,j}_{s-t}(\bx)\right)\,\d s,
\]
where $\Omega^{i,j}$ is the restriction of $\Omega$ to $S^{i,j}$, and $\nabla^{i,j}\cdot\Omega^{i,j}$ is divergence in $S^{i,j}$.
Then, repeated applications of Theorem~\ref{thm:replicator_eu} lead to the conclusion that the solution to \eqref{convective_approximation} is given by
\[
p(\bx,t)=\sum_{i,j}\e^{Q^{i,j}(\bx,t)}p^\ini_{i,j}(\Phi^{i,j}_{-t}(\bx)).
\]
\end{rmk}

\subsection{Peak and average dynamics}
\label{ssec:peak}

We start by showing that the long term dynamics of the average in the Wright-Fisher process, even in the thermodynamical limit, is not governed by the replicator equation. Consider for example, a population of $n$ types, evolving according to
the replicator-diffusion equation with fitness functions given by $\psi^{(i)}:S^{n-1}\to\R$.

From the fact that the final state of the the replicator-diffusion equation is given by equation~(\ref{final_state}),
the coefficients $\pi_i[p^\ini]$, $i=1,\dots,n$ can be calculated in two ways:
\[
\int \rho_i(\mathbf{x})p^\ini(\mathbf{x})\d\mathbf{x}=\int\rho_i(\mathbf{x})p^\infty(\mathbf{x})\d\mathbf{x}=\pi_i[p^\ini]=\int x_ip^\infty(\mathbf{x})\d\mathbf{x}=:\langle p^\infty\rangle_i\ .
\]
Therefore the average of the probability distribution will converge to a certain point of the simplex depending on the initial condition. 
This is completely different from the replicator dynamics, as its solution  converges to a single attractor, periodic orbits, 
chaotic attractors, etc~\citep{HofbauerSigmund}.

Now, we show that the probability distribution concentrates in the ESS; this shows that the peak will behave in
manner similar to the solutions of the replicator dynamics.

Recall that \citep{HofbauerSigmund} an ESS that lies in interior of $\slf^{n-1}$ must be a global attractor of
the replicator equation~(\ref{replicator_ode}). We have then the following result

\begin{thm}\label{thm:delta_conv}
 Assume $p^\ini\in\bm$, $\sing p^\ini\subset\mathop{\mathrm{int}}\slf^{n-1}$ and assume that (\ref{replicator_ode})
has a unique point $\bx^*$ such that for any initial condition $\bx(0)\in\mathop{\mathrm{int}}\slf^{n-1}$, $\lim_{t\to\infty}\bx(t)=\bx^*$.
Then the solution of equation~(\ref{convective_approximation}) is such that
\[
 \lim_{t\to\infty}p(\bx,t)=\delta_{\bx^*}.
\]
\end{thm}

\begin{proof}
Assume, initially, that $\bx^*\in\mathop{\mathrm{int}}\slf^{n-1}$.
Since $\bx^*$ is a globally stable equilibrium for interior initial points, we can find $T>0$, such that, for $t>T$, and sufficiently small $\delta>0$, we have, for any proper compact subset $K\subset \slf^{n-1}$, that:
\[
 \Phi_t(K)\subset B_{\delta}(\bx^*)\subset\mathop{\mathrm{int}}\slf^{n-1}.
\]
where $B_\delta(\bx^*)$ is the open ball of radius $\delta$ and centered at $\bx^*$.

Let $\eta(\bx)$ be a continuous function with support contained in $K$.
Then, for $t>T$, we have that
\[
\int_{\slf^{n-1}}p(\bx,t)\eta(\bx)\,\d\bx = \int_{B_\delta(\bx^*)}p(\bx,t)\eta(\bx)\,\d\bx.
\]
But, let $\epsilon>0$ be given. Since $\eta$ is continuous, possibly with a smaller $\delta>0$, we must have
\begin{equation}
\label{eq:nearid}
\eta(\bx^*)-\epsilon\leq\int_{B_{\delta}(\bx^*)}p(\bx,t)\eta(\bx)\,\d\bx\leq\eta(\bx^*)+\epsilon,
\end{equation}
Now take $(\delta_k,\eps_k)\downarrow0$ such that (\ref{eq:nearid}) is satisfied. This yields a sequence of times $T_k$ such that $T_k\to\infty$ and
\[
\lim_{k\to\infty}\int_{\slf^{n-1}}p(\bx,T_k)\eta(\bx)\,\d\bx = \eta(\bx^*).
\]
Since $\Phi_s(K)\subset\Phi_t(K)$, for $s>t$, the claim follows.

For the case $\bx^*\in\partial\slf^{n-1}$, the result follows from similar arguments, replacing $B_{\delta}(\bx^*)$ by
$B_{\delta}(\bx^*)\cap\slf^{n-1}$.
\qed
\end{proof}

\subsection{Asymptotic approximation}
\label{ssec:local}
Let
\[
 0<\kappa\ll1.
\]
If we perform a regular asymptotic expansion, i.e., if we write $p_\kappa\approx p_0+\kappa p_1+\cdots$, then we find,  for times $t\ll\kappa^{-1}$, that the leading order dynamics is given by

\begin{equation}
\label{eq:loap}
 \partial_{t}p_0+\nabla\cdot(p_0\boldsymbol{\Omega})=0.
\end{equation}

The next theorem shows that this indeed the case, provided we see $p_0$ as the leading order dynamics with respect to the regular part of the probability density.

\begin{thm}\label{thm:conv_to_replicator}
Assume that the fitness are $C^2(S^{n-1})$ functions, and that the initial condition $p^\ini$ is also $C^2(S^{n-1})$. Let $r_\kappa$ be the regular part of the solution of (\ref{replicator_diffusion_eps}), with $\kappa\ge0$. Then $p_0$ is $C^2(S^{n-1})$, and satisfies the conservation law~(\ref{continuous_conservation_laws}). 
Moreover, if $\nabla\cdot\boldsymbol{\Omega}\geq0$, then given $\kappa$ and $K$ positive, there exits a $C$ such that, for $t\ll C\kappa^{-1}$, we have 
\[
 \|r_\kappa(\cdot,t)-p_0(\cdot,t)\|_\infty\leq C\kappa
\]
and
\[
 \|\partial_x^2p_0(\cdot,t)\|_\infty>K
\]
Thus $p_0$ is the leading order asymptotic approximation to $r_\kappa$, for $t\ll\kappa^{-1} C$.
\end{thm}

\begin{proof}
 The statements about $p_0$ follows straightforward by obtaining the solution by the method of characteristics. 

Let $w_\kappa=r_\kappa-p_0$. Then $w_\kappa$ satisfies

\begin{equation*}
\partial_t w_\kappa=\frac{\kappa}{2}\sum_{i,j=1}^{n-1}\partial^2_{ij}\left(D_{ij}w_\kappa\right)-\sum_{i=1}^{n-1}\partial_i\left(\Omega_i w_\kappa\right)
+\frac{\kappa}{2}g_0(\bx,t)
\end{equation*}
with null initial condition, where
\[
 g_0(\bx,t)= \sum_{i,j=1}^{n-1}\partial_{i,j}^2\left(D_{ij}p_0\right) .
\]
Notice that, because of the assumptions on $p^\ini$, we have that $g_0$ is uniformly bounded in time. 

The solution for such a problem is given by Duhammel principle. Let $S(t,t_0)$ be associated  solution operator. We have that 
\[
 w_\kappa(\mathbf{x},t)=\frac{\kappa}{2}\int_0^t S(t,s)g_0(s,x)\d s.
\]
By the maximum principle applied to the semigroup $S(t_2,t_1)$, we have that $\|S(t,s)g_0(s,x)\|\leq M_s$, and by the uniform bound on $g_0$, we have that there exists a constant $M$ such that $M_s\leq M$. Thus, we find that 
\[
 \|S(t,s)g_0(s,x)\|_\infty\leq M.
\]
Hence
\[
 |w_{\kappa}(\bx,t)|\leq \kappa t\frac{M}{2}.
\]
Therefore, taking $C=2M^{-1}$, we find, for $t\ll C\kappa^{-1}$, that:
\[
 \|w_\kappa(t,\cdot)\|_\infty \ll1.
\]
\qed
\end{proof}

\begin{rmk}
 If the condition on $\nabla\cdot\boldsymbol{\Omega}$ is not satisfied, a similar proof shows that if $t\ll-\log(\kappa)$ then
the same conclusion holds. Notice also that this condition is satisfied if the replicator has a globally stable equilibrium in the interior of $S^{n-1}$.
\end{rmk}

\begin{prop}\label{prop:weak_asymp}
Under the same hypothesis of Theorem~\ref{thm:conv_to_replicator}, let $U$ be an open set such that $U\subset S^{n-1}$ and $\bar{U}\cap \partial S^{n-1}=\emptyset$. Then, there exists $C>0$, such that
\[
\left|\int_{U}\left(p_\kappa(\bx,t)-p_0(\bx,t)\right)\d\bx\right|<C\kappa t,
\]
for any $t$.

\end{prop}
\begin{proof}
By Theorem~\ref{thm:conv_to_replicator}, there exists $C'>0$ such that $\|r_\kappa(\cdot,t)-p_0(\cdot,t)\|_\infty<C'\kappa t$, which we write as:
\begin{equation*}
 -C'\kappa t\leq r_k(\bx,t)-p_0(\bx,t)<C'\kappa t.
\end{equation*}
Integrating in $U$ and using  that $\sing(q_\kappa)\cap U=\emptyset$, the result follows.
\end{proof}

Theorem~\ref{thm:delta_conv} shows that, for sufficient large time, the support of the solution of the replicator PDE, equation~(\ref{convective_approximation}),
 will be concentrated in sufficiently small neighbourhoods of $\bx^*$. In particular, this will be true for the maximum. For the replicator-diffusion equation~(\ref{replicator_diffusion_eps}) this
cannot be valid for any value of $\kappa>0$ (as it was proved in Theorem~\ref{thm:final_state}); however,
for strong selection, the initial dynamics given by the replicator-diffusion equation
is similar to
the one given by the replicator ODE.
 This is justified by the following result:

\begin{thm}
Assume that the replicator has a unique global attractor. Then, under the same hypothesis of Theorem~\ref{thm:conv_to_replicator}, we have that given $\eps>0$ and $\delta>0$  there exist a time $t^*$ and a constant $C>0$, depending only on the initial condition, such that
\[
\left|\int_{B_\eps(\bx_*)}p_{\kappa}(\bx,t)\,\d\bx-1\right|<C\kappa t + \delta,
\]
for $t>t^*$.
\end{thm}

\begin{proof}
We have
\[
\left|\int_{B_\eps(\bx_*)}p_{\kappa}(\bx,t)\,\d\bx-1\right|\leq \left|\int_{B_\eps(\bx_*)}\left(p_{\kappa}(\bx,t)-p_0(\bx,t\right)\,\d\bx\right| + \left|\int_{B_\eps(\bx_*)}p_{0}(\bx,t)\,\d\bx-1\right|.
\]
From Theorem~\ref{thm:conv_to_replicator}, we have a constant $C>0$ such that
\[
\left|\int_{B_\eps(\bx_*)}\left(p_{\kappa}(\bx,t)-p_0(\bx,t)\right)\,\d\bx\right|<C\kappa t.
\]
From Theorem~\ref{thm:delta_conv}, we have that there exists a time $t^*$ such that, for $t>t^*$, we have
\[
\left|\int_{B_\eps(\bx_*)}p_{0}(\bx,t)\,\d\bx-1\right|<\delta.
\]
Combining these two calculations yields the result.
\end{proof}

\section{Numerical results}
\label{sec:numerics}

We show, in this section, numerical results for two variants of the Rock-Scissor-Paper game~\citep{HofbauerSigmund}; i.e., fitness are identified with the pay-off
from game theory. In Subsection~\ref{ssec:num_f}, we study the evolution of the discrete evolution numerically in time, and show that the peak of 
distribution behaves accordingly to the replicator equation while the average value of the same distribution converges to a point which is not the ESS.
In Subsection~\ref{ssec:num_b} we obtain explicitly the fixation probability of a given type for the symmetric Rock-Scissor-Paper game. A full animation 
is available in the website indicated in the caption of figure~\ref{simulation}.

\subsection{Forward equation}
\label{ssec:num_f}

\begin{figure}
\includegraphics[width=.32\linewidth]{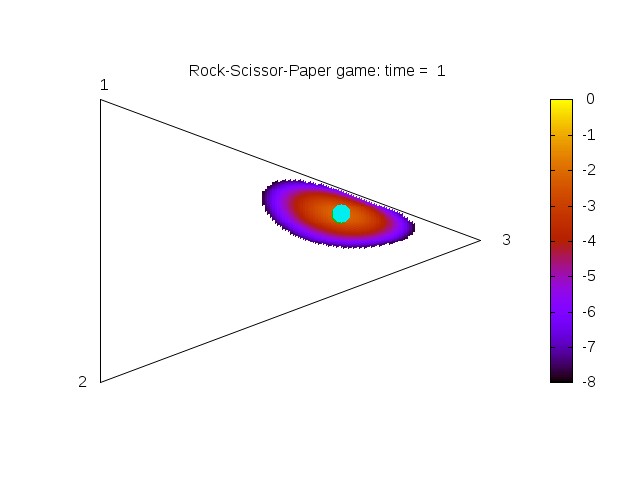}
\includegraphics[width=.32\linewidth]{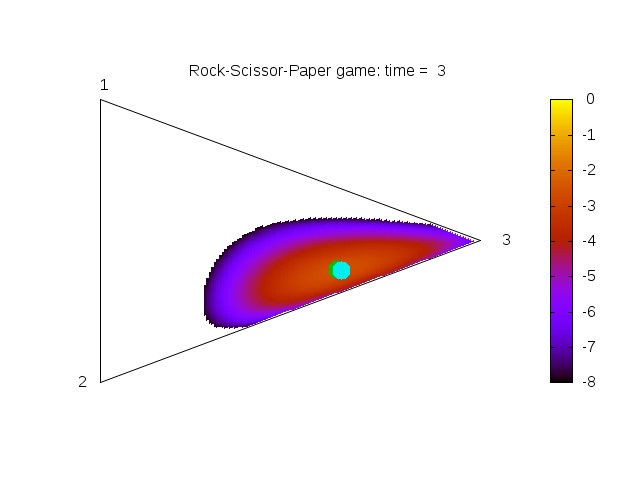}
\includegraphics[width=.32\linewidth]{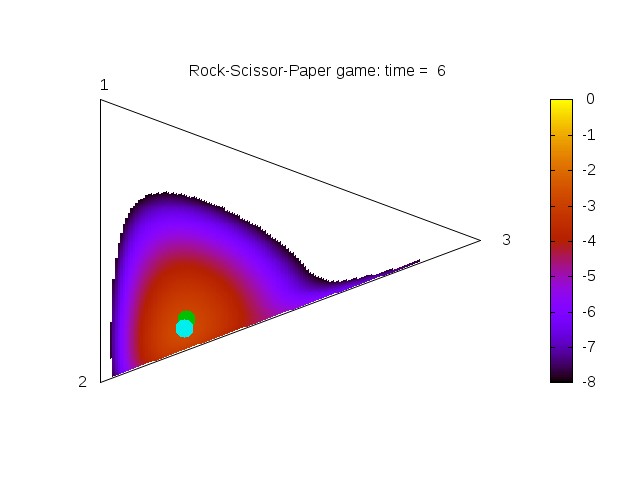}\\
\includegraphics[width=.32\linewidth]{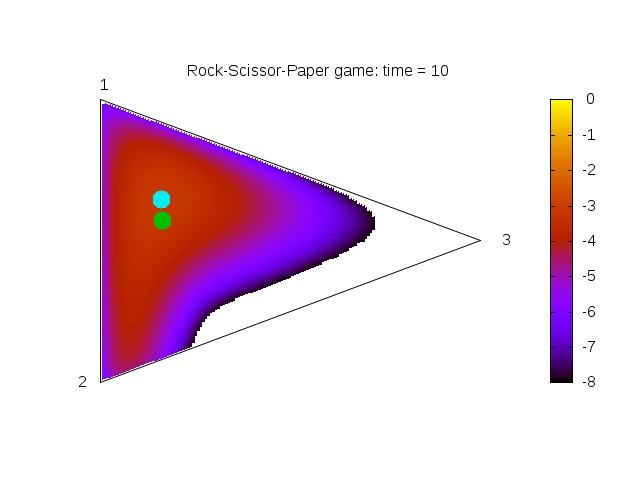}
\includegraphics[width=.32\linewidth]{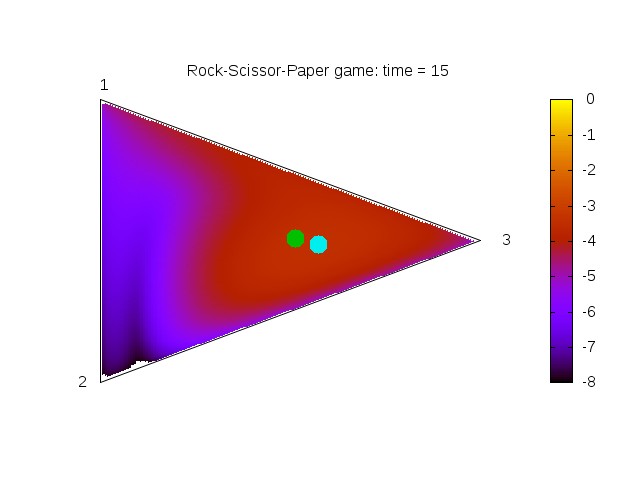}
\includegraphics[width=.32\linewidth]{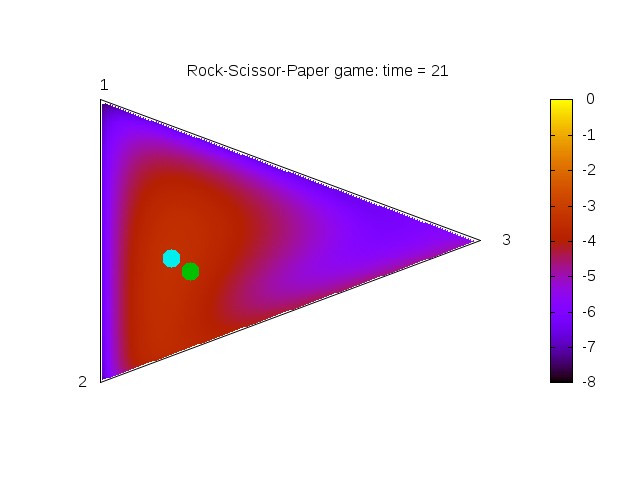}\\
\includegraphics[width=.32\linewidth]{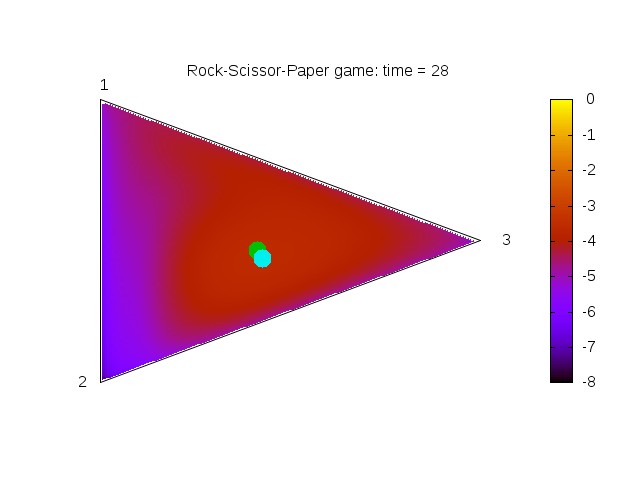}
\includegraphics[width=.32\linewidth]{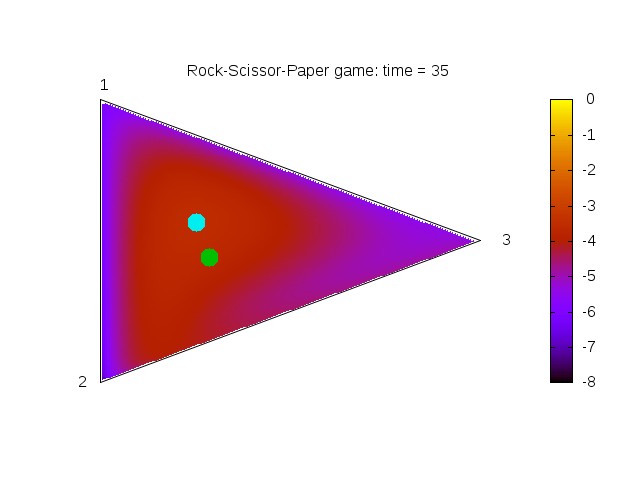}
\includegraphics[width=.32\linewidth]{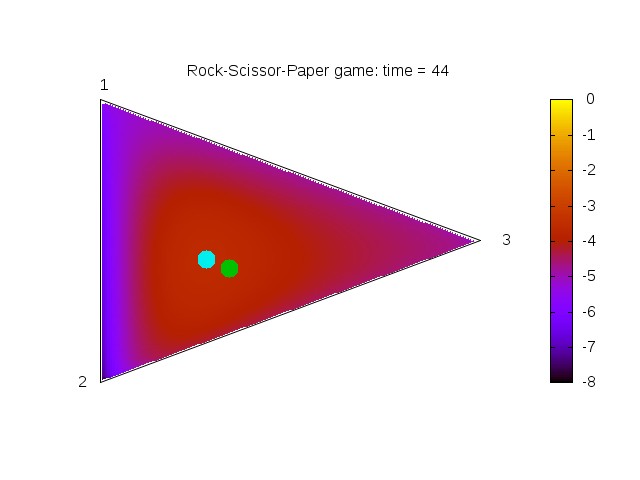}\\
\includegraphics[width=.32\linewidth]{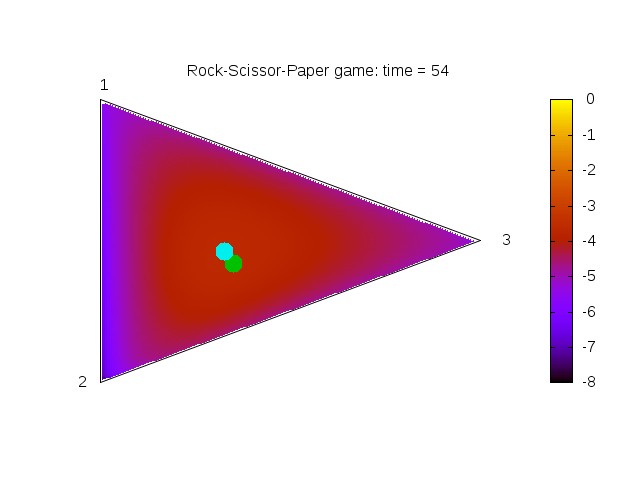}
\includegraphics[width=.32\linewidth]{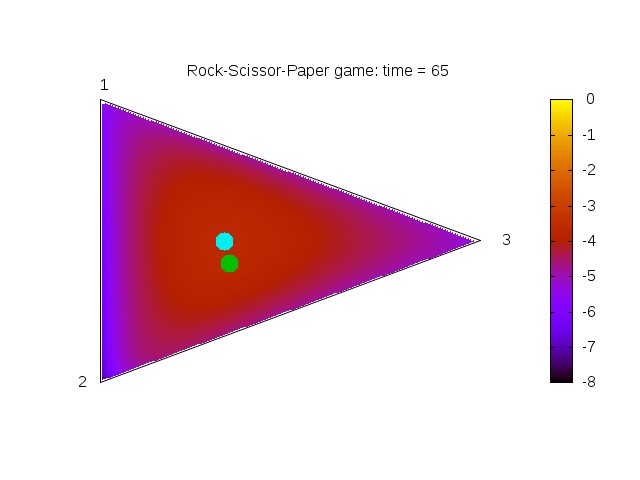}
\includegraphics[width=.32\linewidth]{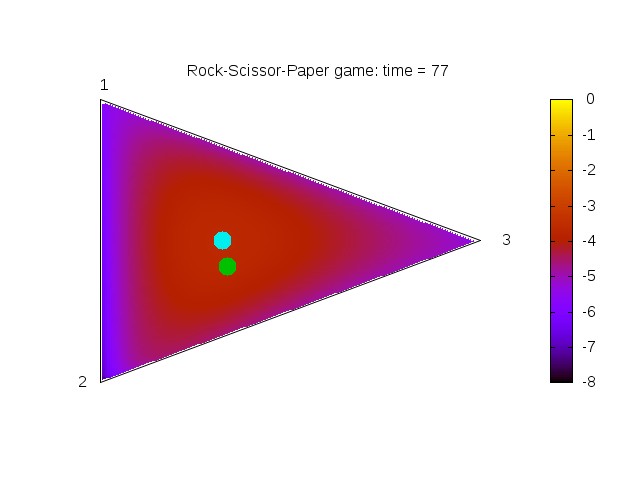}
\caption{Solution for short times (1,3,6,10,15,21,28,35,44,54,65,77) of the Wright-Fisher evolution for a population
of 150 individuals of two given types, with fitness
given by equations~(\ref{fitness_rep}) and (\ref{pay_off_RSP}) for a 
distribution initially concentrated in the interior non-stationary point
$\frac{1}{150}(70,70,10)$. The value of
the distribution $P(x,y,t)$ is in logarithmic scale. Note that the cyan spot, marking
the interior peak of the probability distribution rotates and converges to the 
ESS $\left(\frac{1}{3},\frac{1}{3},\frac{1}{3}\right)$ (along characteristics of the PDE or, equivalently, the trajectories of
the replicator dynamics). At the same time, the green spot marks
the mean value of the probability distribution and also rotates initially. After a long time,
it moves toward its final position, given by 
$\mathbf{x}^\infty\bydef\left(F^{(1)}_{p^\ini},F^{(2)}_{p^\ini},1-F^{(1)}_{p^\ini}-F^{(2)}_{p^\ini}\right)\approx(0.331,0.227,0.442)$.
For a full animation, also for different population sizes $N$, see
{\footnotesize
\texttt{http://dl.dropbox.com/u/11325424/WFsim/RSPFinal.html}}
}\label{simulation}
\label{fig:evolutionRSP}
\end{figure}

We use evolutionary game theory~\citep{JMS,HofbauerSigmund} to define the fitness function.
More precisely, we define a pay-off matrix $\mathbf{M}=\left(M_{ij}\right)_{i,j=1,\cdots,n}$ such that $M_{ij}$ is
the gain (in fitness) of the $i$ type against the $j$ type. 
The fitness of the $i$ type in a population at state $\mathbf{x}$ is 
\begin{equation}\label{fitness_rep}
 \Psi^{(i)}(\mathbf{x})=\sum_{j=1}^nM_{ij}x_j=\left(\mathbf{M}\mathbf{x}\right)_{i}\ .
\end{equation}
In a simulation where both effects, drift and diffusion, are apparent, we have $\mu=\nu=1$ and $\kappa=\bigO(1)$. The last identity implies
$\Delta t=\bigO\left(N^{-1}\right)$. 
Furthermore, from equation~(\ref{WSP}), we have $\psi^{(i)}(\bx)=\frac{1}{\Delta t}\left(\Psi^{(i)}(x)-1\right)$, and therefore, in order to see both effects
we need strong fitness functions and long times, i.e., $\Delta t\approx N^{-1}\ll 1$.

We consider in Figure~\ref{fig:evolutionRSP} the 
evolution of a discrete population of $N=150$ individuals with the pay-off matrix given by
\begin{equation}\label{pay_off_RSP}
 \mathbf{M}=\left(\begin{matrix}
                30&81&29\\
		6&30&104\\
		106&4&30
                  \end{matrix}\right)\ .
\end{equation}
This is know as the generalized Rock-Scissor-Paper game and presents an evolutionary stable state (ESS) $(x^*,y^*,z^*)=\left(\frac{1}{3},\frac{1}{3},\frac{1}{3}\right)$. 
Furthermore, the flow of the replicator dynamics converges in spirals to the ESS. The vertices as well as $(x^*,y^*,z^*)$ are equilibrium points
for the continuum dynamics. See~\cite{HofbauerSigmund} for the choice of values of the matrix $\mathbf{M}$.

Note that the peak moves in inward spirals
around the central equilibrium, following the trajectories of the replicator dynamics, while all the mass diffuses to the boundary.

The green spot indicates the average value for $x$ and $y$; at first it moves in spirals close to the trajectories of the 
replicator dynamics. After a time depending on the value of $N$ it starts to move in the direction of its final point $(x^\infty,y^\infty,z^\infty)
=(\pi_1[p^\ini],\pi_2[p^\ini],\pi_3[p^\ini])$.
This point can be calculated using equation~(\ref{final_state}) and the $n=3$ independent conservation laws. 

\subsection{Backward equation and the decay of the interior $L^1$-norm}
\label{ssec:num_b}

\begin{figure}
\begin{center}
\includegraphics[width=.8\linewidth]{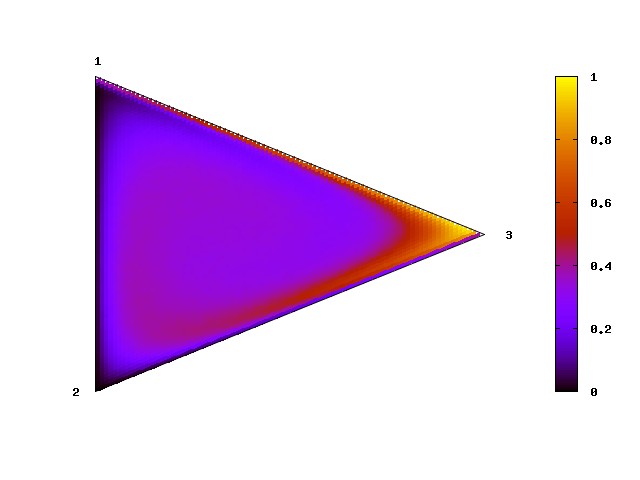}
\end{center}
\caption{Fixation probability of the third type, in a Rock-Scissor-Paper game. This is the numerical solution
of the stationary state 
of the equation (\ref{backward_kimura_full}), simulated by a Wright-Fisher process with $N=150$ and 
pay-off matrix $\left([[20, 0, 40],[40, 20, 0],[0, 40, 20]]\right)$. Note that higher values of the
fixation probability ``rotates'' around the center of the simplex (the stationary state of the
replicator dynamics).
}
\label{fig:backward}
\end{figure}

The stationary state of the backward equation (\ref{backward_kimura_full}) represents the fixation of probability
of a given type. This type is specified by the associated boundary conditions. Let us consider, as an example, that
$n=3$, the evolution is given by the Rock-Scissor-Paper game defined by the matrix
\begin{equation}\label{RSP_matrix}
\mathbf{M}= \left(\begin{matrix}
                   0&40&20\\
		  20&0&40\\
		  40&20&0
                  \end{matrix}
\right)\ ,
\end{equation}
and we study the fixation probability of the third type.
An exact solution is difficult to obtain, as it would be necessary to solve an hierarchy of equations, each solution representing 
boundary conditions of a larger set; however, a numerical solution is extremely easy to compute, as the Wright-Fisher
process is a natural discretisation of the (forward as well as the) backward equation (cf. Theorem~\ref{thm:finite_cons}). This is probably computationally inefficient, and 
different processes can be compatible with the same limit equations. See figure \ref{fig:backward} for
an illustration.

In figure~\ref{fig:l1norm}, we plot the $L^1$ norm in the interior of the simplex and all subsimplexes, showing that that
the probability mass flows from the simplex $S^{n-1}$ to the faces (which are equivalent to the simplexes $S^{n-2}$); the
solution behaves on the faces as the solution of the replicator-diffusion problem with one dimension less. The probability
flows to the ``faces of the faces'', i.e., to simplexes $S^{n-3}$ until it reaches the absorbing state $\mathbf{e}_i$ (simplexes $S^0$) for
$i=1,\dots,n$. We may think of a stochastic process reaching and sticking to the faces of the simplex until they reach their
final spot, the vertices. 
We further observe that, the probability mass in the interior of the simplex is the so-called quasi-stationary distribution of the process, namely, the probability distribution given that it has not been absorbed. See \cite{MeleardVillemonais2011} for a recent survey on the topic.

\begin{figure}
\begin{center}
 \includegraphics[width=0.7\textwidth]{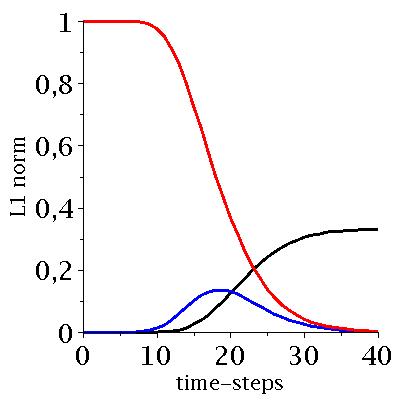}
\caption{Evolution of the probability mass, for the Rock-Scissor-Paper game given by matrix (\ref{RSP_matrix}) and with initial condition concentrated
in the ESS, $p^\ini=\delta_{(\frac{1}{3},\frac{1}{3})}$. The red line indicates the mas ($L^1$-norm) in the interior of the simplex; the blue line, the mass
in the interior any of the faces, and the black line, the mass in any of the vertices.}\label{fig:l1norm}
\end{center}
\end{figure}

\section{Conclusions}
\label{sec:conclusion}

We present a derivation of continuous limits of discrete Markov chain evolutionary models, that are frequency-dependent extensions
the classical Wright-Fisher model, through pure analytical techniques. 
The derivation presented pays close attention to the variety of possible time scalings possible as related to the selection and population size that are measured by two parameters $\mu,\nu\geq1$.  The balance of diffusion and selection ($\mu=\nu=1$ in our terminology) can be seen as slight extension of the results in \cite[Chapter 10]{EthierKurtz} using analytical methods instead of probabilistic arguments, and that favours the forward Fokker-Planck equation instead of the backward. In this sense, from a mathematical point of view, the weak formulation presented for the forward equation seems to be new,  in particular  for the minimal assumptions on the fitness functions. 
The case $\mu>\nu=1$ yields a hyperbolic equation that is the PDE version of the replicator equation. An apparently similar result can be found in  \cite[Chapter 11]{EthierKurtz}, which would correspond formally to take $\Delta t=1$ and $N\gg 1$, without using explicit scaling between these two variables.

With some additional regularity assumptions on the fitnesses functions, we can show that \eqref{weak:replicator_diffusion} is equivalent to \eqref{replicator_diffusion} together with the conservation laws \eqref{continuous_conservation_laws}. In particular, this allows
to characterise the behaviour of $p$ on the lower dimensional subsimplices of $S^{n-1}$. This can be used to obtain equations for the probability of extinction among other information. 

The results here are also related to results
in~\cite{ChampagnatFerriereMeleard_TPB2006,ChampagnatFerriereMeleard_SM2008}, where the idea that the underling scaling influences the macroscopic model was already present, although in a less explicit way than here. 
Nevertheless,  the accelerated birth-death regime in \cite{ChampagnatFerriereMeleard_SM2008} can be seen as a counterpart to our scaling of $\Delta t$ and $N$. On the other hand, the scaling for the fitness are taken as fixed (corresponding to our $\nu=1$ in our terminology), and this explains why they do not obtain the pure diffusive limit in the large population regime. Notice also that the  large population regime taken there seems to annihilate any stochastic effects coming from births and deaths, and that the stochastic effects in this limit are due only to the mutation process.

However, as pointed out above, as we allow more flexibility in the scaling laws, we are able to highlight any of these two factors independently; more precisely, for certain choices of the scaling in the fitnesses functions (namely, 
the exponent $\nu>1$), their influence in the dynamics goes to zero so fast that the limit model is purely diffusive. On the other hand, if we grow the population size fast enough (i.e., $\mu>1$) then
we highlight the deterministic evolution, providing a direct way to compare the replicator equation with the Wright-Fisher process (or, for that matter, also with the Moran process, but, naturally,
in a different time scale). To the best of our knowledge,
this explicit comparison is new. See also~\cite{Fournier_Meleard_AAP2004} for a similar approach.
 
The use of ordinary differential equations in population dynamics is widespread. However, as they are valid only for infinite populations, and 
real populations are always finite, the precise justification of its use and the precise meaning of its solution is seldom made clear.
In this paper, we showed, in a limited framework, but expanding results from previous works~\citep{Chalub_Souza:CMS_2009,Chalub_Souza:TPB_2009}, that ODEs
can be justifiably used to model the evolution of a population. However, the validity of the modelling is necessarily limited in time
(increasing with the population size), and the solution of the differential equation models the most probable state of the system
(therefore, the differential equation would give answers compatible with the maximum likelihood method, but not necessarily compatible
with other estimators).

One of the central issues of the present work is to discuss the possibility of using diffusive approximation for large, but finite, $N$. However, a major challenge to
any one interested to use the replicator-diffusion equation to fit experimental data is the value of $\kappa$.

From the derivation of the replicator-diffusion equation, we see that $\kappa$ is directly linked to the variance of the diffusion, while $\|\psi\|_\infty$ is directly linked with natural selection. Hence, their ratio is an adimensional measure of the relative relevance of the genetic drift with respect to natural selection. If we normalise the fitness functions such that $\|\psi\|_\infty=1$, then $\kappa$ becomes a measure of such relative relevance.

 In this sense, $\kappa^{-1}$ could be an alternative definition of effective population size (see also~\cite{EtheridgeLNM} for usual definitions).
Only when the population is small or times are long in the evolutionary scale, we would expect
order 1 values of $\kappa$.

We are currently applying a
similar technique to epidemiological models; in this case it is necessary to impose boundary conditions in part of the boundary
(as an homogeneous population of infected individual is not stationary, as infected individuals become, with time, removed or
even susceptible) and it is impossible to impose boundary conditions in part of the boundary 
(a population of susceptible remains in this state for ever). Early results were already published in~\cite{Chalub_Souza_2011}. 
The same problem, regarding the imposition of boundary conditions is true if we include mutations in the Moran or Wright-Fisher
model. This is work in progress.

\section*{Acknowledgements}

FACCC was partially supported by CMA/FCT/UNL, financiamento base 2011 ISFL-1-297 and projects PTDC/FIS/101248/2008, PTDC/FIS/70973/2006
from FCT/MCTES/Portugal. FACCC also acknowledges the hospitality of CRM/Barcelona where part of this work was performed and discussions with J. J. Velazquez (Madrid). MOS was partially supported by CNPq grants \#s 309616/2009-3 and 451313/2011-9, and FAPERJ grant \# 110.174/2009. We thank the careful reading and comments of three anonymous referees.

\appendix

\section{Third moment of multinomial distributions}
\label{ap:thirdmomentum}

Let $\balpha$ be a multinomially distributed vector. 
Let also $\partial_i=\frac{\partial\ }{\partial q_i}$ and $|\mathbf{q}|=\sum_i q_i$ (not necessarily equal to 1). Evidently $\partial_iq_j=\delta_{ij}$. 
Then:
\begin{align*}
 \mathbb{E}[\alpha_i\alpha_j\alpha_k]&=\left\{\sum_{\boldsymbol{\alpha}}\alpha_i\alpha_j\alpha_kf(\mathbf{q},\boldsymbol{\alpha},N)\right\}_{|\mathbf{q}|=1}
=\left\{q_i\partial_i\left[q_j\partial_j\left(q_k\partial_k|\mathbf{q}|^N\right)\right]\right\}_{|\mathbf{q}|=1}\\
&=N[q_i\delta_{ij}\delta_{kj}]+N(N-1)\left[q_iq_k\delta_{ij}+q_iq_j\delta_{kj}+q_kq_j\delta_{ki}\right]+N(N-1)(N-2)q_iq_jq_k\ .
\end{align*}

The expression for the third moment now follows from a straightforward
calculation.

\section{Weak formulation in time}
\label{ap:A}

In order to obtain a truly weak formulation, without any requirement upon the regularity of $p$,
we observe that the equation \eqref{eq:weconcludethat} is valid for any time $t_k=t_0+k\Delta t$. Hence, if we also
 let $T=(m+1)\Delta t$ in the equation \eqref{eq:weconcludethat}, and  sum over $k$, we obtain that
\begin{align*}
& \sum_{k=0}^m\sum_{\bx\in S^{n-1}_N}\left(p_N(\bx,t_{k+1})-p_N(\bx,t_k)\right)g(\bx,t_k)\\
&\quad=\frac{1}{2N}\sum_{k=0}^m\sum_{\bx\in S^{n-1}_N}p_N(\bx,t_k)\left(\sum_{i,j=1}^{n-1}x_i(\delta_{ij}-x_j)\partial^2_{ij}g(\mathbf{x},t_k)\right)  \\
&\qquad+\sum_{k=0}^m \left(\Delta t\right)^\nu \sum_{\bx\in S^{n-1}_N}p(\mathbf{x},t_k)\left[\sum_{j=1}^{n-1}x_j\left(\psi^{(j)}(\mathbf{x})-\bar\psi(\mathbf{x})\right)\partial_{j}g(\mathbf{x},t_k)\right]\d\mathbf{x}.
\end{align*}

On summing by parts the left hand side, we obtain
\begin{align*}
 &-\sum_{k=0}^{m-1}\sum_{\bx\in S^{n-1}_N}p_N(\bx,t_k)\left(g(\bx,t_{k+1})-g(\bx,t_k)\right) \\
&\quad- \sum_{\bx\in S^{n-1}}p_N(\bx,t_0)g(\bx,t_0)\,\d\bx + \sum_{\bx\in S^{n-1}_N}p_N(\bx,T)g(\bx,T)\\ 
&\qquad=\frac{1}{2N}\sum_{k=0}^m\sum_{\bx\in S^{n-1}_N}p_N(\bx,t_k)\left(\sum_{i,j=1}^{n-1}x_i(\delta_{ij}-x_j)\partial^2_{ij}g(\mathbf{x},t_k)\right)   \\
&\quad\qquad
+\sum_{k=0}^m \left(\Delta t\right)^\nu \sum_{\bx\in S^{n-1}}p(\mathbf{x},t_k)\left[\sum_{j=1}^{n-1}x_j\left(\psi^{(j)}(\mathbf{x})-\bar\psi(\mathbf{x})\right)\partial_{j}g(\mathbf{x},t_k)\right].
\end{align*}


\end{document}